\documentclass[a4paper,USenglish,cleveref,autoref,thm-restate,numberwithinsect]{lipics-v2021}
\usepackage[bold,full]{complexity}
\usepackage{makecell}
\usepackage{tablefootnote}
\usepackage{xspace}

\usepackage[normalem]{ulem}

\newcommand{\MAdec}[1]{\textbf{MA$_{dec}^{#1}$}\xspace}
\newcommand{\MAbound}[1]{\textbf{MA$_{bound}^{#1}$}\xspace}
\newcommand{\MAsize}[1]{\textbf{MA$_{size}^{#1}$}\xspace}
\newcommand{\MAmin}[1]{\textbf{MA$_{min}^{#1}$}\xspace}
\newcommand{\MAdecL}{\textbf{MA$_{dec}^\logic$}\xspace}
\newcommand{\MAboundL}{\textbf{MA$_{bound}^\logic$}\xspace}
\newcommand{\MAsizeL}{\textbf{MA$_{size}^\logic$}\xspace}
\newcommand{\MAminL}{\textbf{MA$_{min}^\logic$}\xspace}

\newcommand{\MAdecQ}{\textbf{MA$_{dec}^Q$}\xspace}
\newcommand{\MAboundQ}{\textbf{MA$_{bound}^Q$}\xspace}
\newcommand{\MAsizeQ}{\textbf{MA$_{size}^Q$}\xspace}
\newcommand{\MAminQ}{\textbf{MA$_{min}^Q$}\xspace}

\newcommand{\EvalQ}{\textbf{Eval$^Q$}\xspace}
\newcommand{\EvalL}{\textbf{Eval$^\logic$}\xspace}
\newcommand{\Eval}[1]{\textbf{Eval$^{#1}$}\xspace}
\newcommand{\dbSATL}{\textbf{dbSAT$^{\hspace*{1pt}\logic}$}\xspace}
\newcommand{\dbSAT}[1]{\textbf{dbSAT$^{#1}$}\xspace}
\newcommand{\minSC}{\textbf{MIN-SC}\xspace}

\newcommand{\logic}{\mathcal{L}\xspace}
\newcommand{\Var}{\mathbf{Var}}
\newcommand{\Dom}{\mathbf{Dom}}
\newcommand{\var}{\mathbf{var}}
\newcommand{\rel}{\mathbf{rel}}
\newcommand{\free}{\mathbf{free}}
\newcommand{\bound}{\mathbf{bound}}
\newcommand{\Body}{\mathbf{Body}}
\newcommand{\arity}{\mathbf{arity}}
\newcommand{\df}{~\mathrel{{:}{-}}~}

\newcommand{\adom}{\textbf{adom}}
\newcommand{\bfE}{\mathbf{E}}
\newcommand{\bfI}{\mathbf{I}}
\newcommand{\bfJ}{\mathbf{J}}
\newcommand{\bfR}{\Sigma}
\newcommand{\bfS}{\mathbf{S}}
\newcommand{\repI}{\mathbf{I}_\Rep}

\newcommand{\Facts}{\mathbf{Facts}}
\newcommand{\ans}{\mathbf{ans}}
\newcommand{\Unfoldings}{\mathbf{Unfoldings}}

\newcommand{\Img}{\textrm{Im}}
\newcommand{\rng}{\textrm{rng}}
\newcommand{\card}[1]{\lvert #1 \rvert}
\newcommand{\tup}[1]{\overline{#1}}
\newcommand{\len}{\textrm{len}}

\newcommand{\Pred}{\leq_\P}
\newcommand{\cost}{\textrm{cost}}
\newcommand{\opt}{\textrm{opt}}
\newclass{\FEXP}{FEXP}
\newclass{\twoEXP}{2EXP}
\newclass{\FtwoEXP}{F2EXP}
\newclass{\twoNEXP}{2NEXP}

\newcommand{\Rep}{\mathsf{U}}

\newcommand{\Ins}{\mathsf{Ins}}
\newcommand{\Del}{\mathsf{Del}}

\newcommand{\boundL}{\beta_{\logic}}
\newcommand{\cbound}{c_Q}

\newcommand{\RA}{\textrm{\upshape\bf RA}\xspace}
\newcommand{\CQ}{\textrm{\upshape\bf CQ}\xspace}
\newcommand{\CQneg}{\textrm{\upshape\bf CQ$_\lnot$}\xspace}
\newcommand{\CQsjf}{\textrm{\upshape\bf CQ$_\lnot$[sjf]}\xspace}

\newcommand{\UCQ}{\textrm{\upshape\bf UCQ}\xspace}
\newcommand{\UCQneg}{\textrm{\upshape\bf UCQ$_\lnot$}\xspace}
\newcommand{\UCQsjf}{\textrm{\upshape\bf UCQ$_\lnot$[sjf]}\xspace}
\newcommand{\DL}{\textrm{\upshape\bf Datalog}\xspace}

\newcommand{\spDL}{\textrm{\upshape\bf sp-Datalog}\xspace}
\newcommand{\DLneg}{\textrm{\upshape\bf Datalog$_\lnot$}\xspace}
\newcommand{\negCQ}{\textrm{\upshape\bf CCQ}\xspace}

\newcommand{\PJ}{\textrm{\upshape\bf PJ}\xspace}

\newcommand{\SJUneg}{\textrm{\upshape\bf SJU$_\lnot$}\xspace}
\newcommand{\SJneg}{\textrm{\upshape\bf SJ$_\lnot$}\xspace}
\newcommand{\SPUneg}{\textrm{\upshape\bf SPU$_\lnot$}\xspace}
\newcommand{\SPneg}{\textrm{\upshape\bf SP$_\lnot$}\xspace}
\newcommand{\SPJUneg}{\textrm{\upshape\bf SPJU$_\lnot$}\xspace}

    \pdfoutput=1 
    \hideLIPIcs  

\title{The Complexity of Finding Missing Answer Repairs}
\author{Jesse Comer}{University of Pennsylvania, United States \and \url{https://jessecomer.github.io/}}{jacomer@seas.upenn.edu}{https://orcid.org/0009-0006-9734-3457}{}
\author{Val Tannen}{University of Pennsylvania, United States \and \url{https://www.cis.upenn.edu/~val/home.html}}{val@seas.upenn.edu}{0009-0008-6847-7274}{}
\authorrunning{J. Comer and V. Tannen}
\Copyright{Jesse Comer and Val Tannen}



\nolinenumbers 

\EventEditors{Balder ten Cate and Maurice Funk}
\EventNoEds{2}
\EventLongTitle{29th International Conference on Database Theory (ICDT 2026)}
\EventShortTitle{ICDT 2026}
\EventAcronym{ICDT}
\EventYear{2026}
\EventDate{March 24--27, 2026}
\EventLocation{Tampere, Finland}
\EventLogo{}
\SeriesVolume{365}
\ArticleNo{12}

\begin{document}

\maketitle

\begin{abstract}
We investigate the problem of identifying database repairs for missing tuples in query answers. We show that when the query is part of the input - the \emph{combined complexity} setting - determining whether or not a repair exists is polynomial-time is equivalent to the satisfiability problem for classes of queries admitting a weak form of projection and selection. We then identify the sub-classes of unions of conjunctive queries with negated atoms, defined by the relational algebra operations permitted to appear in the query, for which the minimal repair problem can be solved in polynomial time. In contrast, we show that the problem is $\NP$-hard, as well as set cover-hard to approximate via strict reductions, whenever both projection and join are permitted in the input query. Additionally, we show that finding the size of a minimal repair for unions of conjunctive queries (with negated atoms permitted) is $\OptP[\log(n)]$-complete, while computing a minimal repair is possible with $O(n^2)$ queries to an $\NP$ oracle. With recursion permitted, the combined complexity of all of these variants increases significantly, with an $\EXP$ lower bound. However, from the \emph{data complexity perspective}, we show that minimal repairs can be identified in polynomial time for all queries expressible as semi-positive datalog programs.
\end{abstract}

\section{Introduction}
\label{sec:intro}
We study the computational complexity of the following \emph{missing answer repair problem}: given a query $Q$, a database instance $\bfI$ (the \emph{input} or \emph{source} database) and a tuple $\tup{a}$ which may not be in the answer $Q(\bfI)$, the problem asks to find an update $\Rep$ to $\bfI$ (a collection of insertions and deletions of tuples in $\bfI$, collectively called a \emph{repair}), producing a \emph{repaired} instance $\repI$ such that $\tup{a}\in Q(\repI)$. The notion of ``finding'' a repair can be made precise in four ways: deciding if a repair exists ($\MAdec{}$), deciding if a repair of size at most $k$ exists ($\MAbound{}$), determining the size of a minimal repair ($\MAsize{}$), and computing a repair of minimum size ($\MAmin{}$). For each version of the problem, we consider both \emph{combined complexity}, where the problem is parametrized by a class of queries, and \emph{data complexity}, where the problem is parametrized by a single fixed query. In both cases, we consider the schema of the database to be fixed. The missing answer problem is related to two concepts already studied in some depth in the literature: \emph{view updates} and \emph{database repairs}. We outline both and explain how our results differ from previous ones. We also discuss how negation relates the missing answer repair problem to other similar problems in the literature.

\subparagraph*{View updates.}
The missing answer repair problem has applications to the \emph{view update problem}~\cite{DayalB82}: if $Q$ defines a view, then the missing answer repair problem is to find an update to the input database instance so that a given tuple is \emph{inserted} in the view defined by $Q$ on the updated instance. To the best knowledge of the present authors, most \emph{complexity} studies of updating views have focused so far on the problem of finding instance updates in order to \emph{delete} a tuple from a view. This previous work was restricted to monotonic queries, and so updates were limited to deletions of tuples from the input database (see the work on \emph{deletion propagation}~\cite{BunemanKT02, CongFGLL12, KimelfeldVW12, Kimelfeld12, Miao2016complexity} and on \emph{resilience}~\cite{FreireGIM15, FreireGIM20, MakhijaG23}). In contrast, we focus on view insertions, rather than view deletions, and we consider more general (non-monotonic) queries. Consequently, our updates are comprised of both insertions and deletions. Furthermore, our results are much more optimistic than those for deleting tuples from queries. In contrast to resilience, whose data complexity is $\NP$-hard even for some conjunctive queries, we show here that the missing answer repair problem is solvable in polynomial time, in data complexity, for all semi-positive datalog programs. 

Work has also been done on \emph{insertion propagation}, which seeks to add tuples to views, and \emph{annotation propagation}, which seeks to modify an annotation of a tuple in the source database in order to modify the annotation of a target tuple in the view \cite{BunemanKT02, Tan2004containment, Cong2006annotation, CongFGLL12}. Previous work on these problems differs in important ways from the missing answer problem studied here. To understand these differences, note that, in general, updates to a database instance with the goal of updating a view are not unique. As a result, the work on view updates has always emphasized finding instance updates which satisfy some desired criteria~\cite{DayalB82}. As mentioned previously, we focus on making insertions of tuples to query results while minimizing the size of the update applied to the database. In our later discussion of related work, we outline alternative criteria for updates that have been considered.

\subparagraph*{Database repairs.}
The notion of a database \emph{repair} originates in the literature on consistent query answering (CQA)~\cite{ArenasBC99}. In that context, given a set $\IC$ of integrity constraints and an inconsistent database instance $\bfI$ (one not satisfying $\IC$), a repair is a minimally-modified consistent database instance $\bfI'$. The goal of CQA is to identify whether or not a given tuple is a \emph{certain answer}, i.e., in the result of a given query $Q$ for \emph{all} repairs $\bfI'$. There are a number of ways to define what it means for $\bfI'$ to be ``minimally-modified'' from $\bfI$. Most similar to the notion used here are the \emph{cardinality repair semantics}~\cite{arenas2003answer,lopatenko06complexity}. Given an inconsistent database instance $\bfI$ and consistent database instances $\bfI'$ and $\bfI''$, write $\bfI' \leq^\bfI_C \bfI''$ if $\card{\bfI \oplus \bfI'} \leq \card{\bfI \oplus \bfI''}$ (where $S \oplus T$ denotes the symmetric difference of sets $S$ and $T$). In CQA with cardinality repairs, the goal is to find tuples which appear in $Q(\bfI')$ for \emph{all} $\leq^\bfI_C$-minimal consistent databases $\bfI'$. In other words, minimality for cardinality repairs is defined by the size of the symmetric difference between the input database and the modified database.

CQA under cardinality repair semantics is different from the missing answer repair problem in a number of important ways. A \emph{repair} in CQA is a $\leq^\bfI_C$-minimal database which is consistent with the integrity constraints $\IC$ in the input. In contrast, a repair for us is any update which places the input tuple in the answer to the input query on the updated database. This means that, given an input database and input query, the minimal repair size depends on the input tuple as well. Intuitively, rather than fixing the repair size and looking for tuples appearing in the answer to the query on all repairs, we fix the tuple and look for the smallest repair. Due to these differences, it is not immediately clear how to relate the computational complexity of these two problems.

Another context in which database repairs appear is \emph{data cleaning}~\cite{Fan2012foundations}. One of the aspects of cleaning data collected, for example, from the Web, requires restoring expected integrity constraints, and hence repairing inconsistent database instances (e.g.~\cite{Kolahi2009approximating, XuZhaAlaTan18, Bertossi2018measuring, Carmeli2024database, GraedelTan24}). This is similar to CQA, except that the priority is only to restore consistency. If we consider a query language capable of expressing certain integrity constraints (typically as the truth of a Boolean query), then the missing answer repair problem asks for a minimal cardinality repair in order to update an inconsistent database to a consistent one. We do not view the work in the present paper as particularly relevant to this line of research, since expressing commonly studied forms of integrity constraints, such as functional and inclusion dependencies, or tuple-generating dependencies, is not possible without using some form of universal quantification in the query. In contrast, we focus here on fragments of semi-positive datalog, which lacks universal quantification\footnote{However, it should be noted that some integrity constraints, such as tuple-generating dependencies, can be expressed in semi-positive datalog when the universally-quantified variables are instantiated with constants. In this sense, our results show that tuple-generating dependences can be repaired ``tuple-by-tuple'' in polynomial-time (for each tuple) data complexity.}.

\subparagraph*{Missing answers and negation.}
As deletion propagation and resilience are concerned with modifying a database to add tuples to the answer of a query, they might be aptly interpreted as \emph{wrong answer} repair problems. Missing and wrong answer repairs are related as follows. Given a query $Q$, let $\lnot Q$ (the \emph{negation} of $Q$; see Remark \ref{remark:negation}) denote a query which, on each database instance $\bfI$, returns the complement of $Q(\bfI)$ (with respect to all tuples of appropriate arity over the active domain of $\bfI$). Clearly, the problem of deleting a tuple from the answer to $Q$ is the same as inserting a tuple into the answer to $\lnot Q$. Prior work on wrong answer repairs has focused on non-recursive monotone queries, primarily conjunctive queries or unions of conjunctive queries \cite{Cong2006annotation, CongFGLL12, FreireGIM15, Miao2018aggregation, Miao2018complexity, Miao2020functional, Miao2020results}. In contrast, the present paper focuses on semi-positive datalog and its fragments. While negation provides a direct link between missing and wrong answer repairs, the application of negation to unions of conjunctive queries yields \emph{universal} (rather than \emph{existential}) queries which are not expressible in the sub-classes of semi-positive datalog that we consider here.

\subparagraph*{Contributions.}
We consider four versions of the problem in the combined complexity setting: $\MAdecL$, $\MAboundL$, $\MAsizeL$, and $\MAminL$, where the problem is parametrized by a class $\logic$ of queries. Our main combined complexity results are summarized in the following table.

\begin{table}
\centering
\begin{tabular}{|c|c|c|c|c|}
\hline
\textbf{Query Class} & \MAdecL & \MAboundL & \MAsizeL & \MAminL \\
\hline
$\UCQsjf$   & $\P$
            & $\P$
            & $\P$
            & $\P$ \\
\hline
$\SJUneg$   & $\P$ 
            & $\P$ 
            & $\P$
            & $\P$ \\
\hline
$\SPUneg$   & $\P$
            & $\P$
            & $\P$
            & $\P$ \\
\hline
$\PJ$       & $\P$
            & $\NP$-compl.
            & $\OptP[\log(n)]$-compl.
            & $\FP^{\NP[n^2]}$ \\
\hline
$\UCQneg$   & $\P$
            & $\NP$-compl.
            & $\OptP[\log(n)]$-compl.
            & $\FP^{\NP[n^2]}$ \\
\hline
$\DL$       & $\EXP$-compl.
            & $\EXP$-compl.
            & $\FEXP$
            & $\FEXP$ \\
\hline
$\spDL$     & $\EXP$-compl.
            & $\twoEXP$
            & $\FtwoEXP$
            & $\FtwoEXP$ \\
\hline
$\RA$       & Undecidable
            & $\twoNEXP$
            & Uncomputable
            & Uncomputable \\
\hline
\end{tabular}
\caption{Combined complexity results. We write $\UCQsjf$ to denote the class of self-join-free $\UCQneg$ queries. Note that $\SJUneg$ (resp. $\SPUneg$) denotes the class of projection-free (resp. join-free) $\UCQneg$ queries. Similarly, $\PJ$ denotes the class of selection-free conjunctive queries. The tag ``sp'' for Datalog stands for ``semi-positive'' (i.e., programs having atomic negation). We formally define these classes in Section \ref{sec:prelim}.}
\label{table:combined}
\end{table}

Some of the results in Table \ref{table:combined} are not tight; in cases like this, the table depicts only our best upper bound for the complexity of the problem. For example, the upper bound we provide for $\MAminL$ for $\logic \in \{ \PJ, \UCQneg \}$ is $\FP^{\NP[O(n^2)]}$, which does not match our lower bound of $\FP^{\NP[O(\log(n))]}$. However, this is typical for such results on oracle complexity for (the function versions of) optimization problems\footnote{The formal definitions of $\FP^{\NP[O(\log(n))]}$ and $\FP^{\NP[O(n^2)]}$ can be found in Section~\ref{sec:prelim}.}. For example, to the authors' knowledge, the best known lower bound for constructing optimal solutions for optimization problems like minimum set cover, minimum vertex cover, maximum clique, and others, is $\FP^{\NP[O(\log(n))]}$, while the tightest known upper bound known for these problems is $\FP^{\NP[O(n)]}$. Another bound which is not tight is the complexity of $\MAbound{\spDL}$. We show in Section \ref{sec:combined} that this problem is $\EXP$-hard and contained in $\twoEXP$, but we leave open its precise complexity.

In addition to the results depicted in the Table~\ref{table:combined}, we also provide a minimum set cover-hardness of approximation result for $\MAmin{\PJ}$ via strict reductions (a special case of the $L$-reductions defined in \cite{PapadimitriouY91}; see Definition \ref{def:strict-reduction}). The class $\OptP$, introduced in \cite{Krentel88} and expounded further in \cite{GasarchKR95}, is used to capture the complexity of determining the optimal value for $\NP$-complete optimization problems. Given the $\NP$-completeness of $\MAbound{\PJ}$, this is the natural class to capture the complexity of $\MAsize{\PJ}$.

In data complexity, we study the $\MAdecQ$, $\MAboundQ$, and $\MAminQ$ problems, where the query $Q$ is fixed; the results are summarized in Table \ref{table:data}. Note that the complexity of the $\MAboundQ$, $\MAsizeQ$, and $\MAminQ$ problems for $\RA$ queries $Q$ varies from query to query; we discuss the decidability of these problems at the end of Section \ref{sec:data}. Furthermore, the results on $\negCQ$ follow immediately from the observations in the introduction involving negation, together with results from \cite{FreireGIM15, FreireGIM20} providing dichotomies for resilience of certain sub-classes of conjunctive queries. The $\MAsizeQ$ and $\MAminQ$ results for this class follow from standard arguments (\cite[pgs. 415-422]{papadimitriou2003computational}) using binary search with $\NP$ oracle queries to identify the size of an optimal repair as well as a polynomial-time construction of a minimal repair via $\NP$ oracle queries. We include this row to indicate that the various missing answer problems can be intractable, even in data complexity, for reasonable classes of queries.

\begin{table}
\centering
\begin{tabular}{|c|c|c|c|}
\hline
\textbf{Query Class} & \MAboundQ & \MAsizeQ & \MAminQ \\
\hline
$\spDL$     & $\P$
            & $\P$
            & $\P$ \\
\hline
$\negCQ$    & $\P$ or $\NP$-compl.
            & $\P$ or $\FP^{\NP[\log(n)]}$
            & $\P$ or $\FP^{\NP}$ \\
\hline
$\RA$       & Decidable
            & Computable
            & Computable \\
\hline
\end{tabular}
\caption{Data complexity results. We write $\negCQ$ to denote the class of queries of the form $\neg Q$ (see Remark \ref{remark:negation}), where $Q$ is a conjunctive query.}
\label{table:data}
\end{table}

\subparagraph*{Related work.}
Under combined complexity, the missing answer repair problem generalizes \emph{query satisfiability}~\cite{abiteboul1995foundations} (recall that a query $Q$ is satisfiable if there exists a database instance $\bfI$ where $Q(\bfI) \neq \emptyset$). More precisely, a Boolean query $Q$ is satisfiable if and only if there exists a repair for the input $\langle Q, \bfE, \epsilon \rangle$ to the missing answer decision problem, where $\bfE$ and $\epsilon$ are the empty database and empty tuple, respectively. In Section \ref{sec:combined}, we extend this observation to non-Boolean queries, and we use this relationship to transfer known complexity lower bounds for satisfiability to the missing answer problem for many classes of queries.

As mentioned in the introduction, there are, in general, many updates to a database which might yield a desired change in a view. Consequently, work on view updates has focused on finding updates with certain goals in mind: (1) minimizing the size of the instance update, (2) updating the instance while preserving specified integrity constraints, (3) determining if an update to the view exists which does not induce any other view ``side-effects,'' and (4) minimizing ``side-effects'' on the view. Note that \emph{view side-effects} refer to additional changes to the tuples appearing in the view beyond the original desired update.

We focus on view insertions which minimize the number of insertions and deletions made to the database (i.e., goal (1) above). This problem has been studied to some extent in prior work. In particular, it was shown in \cite{Cong2006annotation, CongFGLL12} that for selection-free conjunctive queries, under combined complexity, computing a minimal repair for view insertions is $\NP$-hard, while it can be done in polynomial-time for unions of conjunctive queries which lack either joins or projections. We refine this result in several ways. First, we provide a similar hardness result for selection-free conjunctive queries, even under the assumption that the schema is fixed. In contrast, \cite{Cong2006annotation} allows for arbitrarily many relations to occur in the schema, and their proof of hardness depends on this fact. Second, we show that the polynomial-time results are extendable to queries with atomic negation. Finally, we provide concrete upper bounds in $\OptP$ and $\FP^{\NP}$ for the functional versions of these problems ($\MAsize{}$ and $\MAmin{}$). Other work on insertion propagation has emphasized (2) and (3), particularly on computing updates to the source database which maintain integrity constraints such as functional dependencies without view side-effects \cite{Miao2016complexity, Miao2020functional}. This has also been the focus for much of the work on annotation propagation \cite{Cong2006annotation, CongFGLL12}. In the study of deletion propagation and resilience, all of (1) through (4) have been studied \cite{BunemanKT02, Cong2006annotation, CongFGLL12, KimelfeldVW12, Kimelfeld12, FreireGIM15, Miao2018complexity, Miao2018aggregation, Miao2020results, FreireGIM20, Miao2023deletion}, albeit not for query languages as powerful as those we study in this paper.

In \cite{Miao2018complexity}, the authors define a \emph{bounded insertion propagation problem}, which they show is $\Sigma^P_2$-complete for conjunctive queries under combined complexity. The problem defined in that paper is analogous to our definition of the $\MAbound{}$ problem, but differs in important ways. First, the schema is not fixed. Second, the relations are assumed to be associated with typed attributes, where the domain of each attribute is specified in the input. In contrast, we assume that the schema is fixed, and that any data (even outside the active domain of the input structure) can appear in repairs. Both of these differences feature prominently in the $\Sigma^P_2$-completeness proof in \cite[Theorem 4]{Miao2018complexity}, and so it is clear that the problem is computationally much different. In particular, a consequence of our results is that the $\MAbound{}$ problem for conjunctive queries is $\NP$-complete.

One significant difference between the work presented here and prior work on the missing and wrong answer problems is that we consider more powerful query languages. We consider conjunctive queries, unions of conjunctive queries, and datalog, as well as their extensions with negated (extensional) atoms. In particular, such extensions require that repairs consist of both insertions and deletions, as opposed to just insertions, which suffice only to repair missing answers for monotone queries. With respect to combined complexity, the relationship with satisfiability is for arbitrary classes of queries closed only under composition with weak forms of selection or projection. In contrast, previous work on wrong answers focuses on unions of conjunctive queries (UCQs) \cite{BunemanKT02, CongFGLL12} or just conjunctive queries (CQs) \cite{KimelfeldVW12, Kimelfeld12, FreireGIM15, FreireGIM20}. Another significant difference is that we allow for insertions with data that is not necessarily in the \emph{active domain} of the input database $\bfI$ or the query $Q$. For queries without negation, this makes no difference; however, as mentioned above, we consider queries with negated atoms. In contrast, prior work on missing answers for arbitrary first-order queries~\cite{GraedelTan24,XuZhaAlaTan18} limits insertions to those constructed from active domain elements.

It is also important to highlight a distinct approach to view updates originating in \cite{DayalB82}, where the authors define \emph{update translations} and provide conditions for a \emph{correct} translation of view updates. This work is extended in \cite{bancilhon1981update}, which introduces the notion of the \emph{complement} of the view,  which is a particular type of mapping from database instances to pairs of views. Note that the translation and complement definitions are independent of a given database model. In \cite{cosmadakis1984updates}, the authors apply this framework specifically to the relational database model, in the simplified setting with a single relation, functional dependencies, and views defined by projections on that relation. The authors provide several complexity results for decision problems related to view insertions and view complements. In general, these problems appear in the lower levels of the polynomial hierarchy, and more work is needed to determine the relationship of this framework to the missing answer repair problem studied here.

The issue of finding \emph{explanations} for missing answers has also been studied using a variety of techniques~\cite{ChapmanJ09,HerschelHT09, HerschelH10,TranC10,WuM13,RoyS14,BidoitHT14a,BidoitHT14b,RoyOS15}, and especially using formalizations of provenance \cite{BunemanKT01,GreenKT07,HuangCDN08,GraedelTan24,XuZhaAlaTan18} and of causality. The latter work discusses two ways of quantifying the relative importance of causes: \emph{responsibility}~\cite{MeliouGHKMS10,MeliouGMS10} and \emph{resilience}~\cite{FreireGIM15,FreireGIM20,MakhijaG23}. In fact, computing resilience is essentially computing repairs for wrong answers.

\section{Preliminaries}
\label{sec:prelim}
\subparagraph*{Basic notation and definitions.}
We fix countably infinite sets $\Var$ and $\Dom$ of \emph{variables} and \emph{constants}, respectively. \emph{Terms} are elements of $\Var \cup \Dom$. An \emph{assignment} is a partial map $g: \Var \to \Dom$. We write $x,y,z,u,v$ to denote variables, $a,b,c,d,e$ to denote constants, and $s,t$ to denote arbitrary terms. We use an overline notation to denote tuples of variables or constants, and we write $\len(\tup{a})$ to denote the length of a tuple $\tup{a}$. Given a tuple $\tup{t}$ of terms and an assignment $g$, we write $g(\tup{t})$ to denote the tuple of constants obtained by replacing all variables in $\tup{t}$ with their image under $g$. We will sometimes abuse notation by treating tuples as sets. We fix a finite schema $\bfR$ of \emph{extensional database predicates} (EDBs) and an infinite schema $\bfS$ of \emph{intensional database predicates} (IDBs). Each relation symbol $F \in \bfR \cup \bfS$ has an associated natural number \emph{arity} (denoted $\arity(F)$), and we assume $\bfS$ contains infinitely-many symbols of each arity.

\emph{Atoms} are either equalities of the form $s=t$, where $s$ and $t$ are terms, or expressions of the form $F(t_1,\hdots,t_k)$, where $F \in \bfR \cup \bfS$ with $\arity(F) = k$ and $t_1,\hdots,t_k$ are terms. An \emph{extensional} (resp. \emph{intensional}) atom is one constructed with a relation symbol from $\bfR$ (resp. $\bfS$). A \emph{positive literal} is an atom, and a \emph{negative literal} is an expression of the form $\lnot \alpha$, where $\alpha$ is a positive literal. We write $\alpha,\beta,\gamma$ to denote (positive or negative) literals. We write $\rel(\alpha)$ to denote the relation symbol occurring in a literal $\alpha$ and $\var(\alpha)$ to denote the variables occurring in $\alpha$. A \emph{fact} is an extensional atom with no variables. Given a subset $D \subseteq \Dom$, we write $\Facts(D)$ for all facts containing only constants in $D$. An \emph{instance} is a finite set of \emph{facts}; we denote instances by $\bfI$, $\bfJ$.

A \emph{datalog rule} is an expression $r$ of the form
\begin{equation}
\label{eq:rule}
S(\tup{x}) \df \alpha_1,\hdots,\alpha_n,
\end{equation}
where $S(\tup{x})$ is an intensional atom, $\tup{x}$ is a tuple over $\Var$ with $\len(\tup{x}) = \arity(S)$, each $\alpha_i$ is a positive literal with $\rel(\alpha_i) \in \bfR \cup \bfS$, and $\tup{x} \subseteq \tup{y} = \bigcup_{i \leq n} \var(\alpha_i)$. We refer to $S(\tup{x})$ as the \emph{head} of $r$, and $\alpha_1,\hdots,\alpha_n$ as the \emph{body} of $r$, denoted $\Body_r(\tup{y})$. While we do not allow constants to appear in the head of rules, these can be simulated using equality atoms in the body of the rule. We write $\free(r)$ to denote the variables occurring in the head of $r$ and $\bound(r)$ to denote the variables in the body of $r$ which do not occur in $\free(r)$. Additionally, we set $\var(r) = \free(r) \cup \bound(r)$. If $\tup{t}$ is a tuple of terms with $\len(\tup{y}) = \len(\tup{t})$, then we write $\Body_r(\tup{t})$ to denote the expression obtained by substituting the variables in $\tup{y}$ for those in $\tup{t}$.

A \emph{datalog program} $P$ is a finite sequence $r_1, \hdots, r_m$ of rules such that every intensional relation symbol which occurs in the body of some $r_i$ also must occur in the head of some $r_j$; each datalog program $P$ also has a designated intensional relation symbol $\ans_P$ denoting the \emph{answer predicate} for $P$. The \emph{arity} of a program $P$ (denoted by $\arity(P)$) is the arity of its answer predicate. We also assume that all datalog programs are \emph{safe}, meaning that every variable in each rule appears in some positive literal in the body of the rule.

A \emph{conjunctive query} is a datalog program containing a single rule whose body contains only extensional atoms, and a \emph{union of conjunctive queries} is a datalog program consisting of conjunctive queries, each having the same head. We write $\CQ$, $\UCQ$, and $\DL$ to denote the classes of conjunctive queries, unions of conjunctive queries, and datalog programs, respectively. A \emph{semi-positive datalog rule} has the same form as a datalog rule (see Equation \ref{eq:rule}), except that we allow each $\alpha_i$ to be a positive literal with $\rel(\alpha_i) \in \bfR \cup \bfS$, a negative literal with $\rel(\alpha_i) \in \bfR$, or an \emph{inequality atom} of the form $s \neq t$, where $s,t \in \Dom \cup \Var$. We write $\CQneg$, $\UCQneg$, and $\spDL$ for the extensions of $\CQ$, $\UCQ$, and $\DL$ with negated extensional atoms and inequality atoms. We will generally use $P$ to denote $\DL$ or $\spDL$ programs, and $Q$ to denote queries in $\UCQneg$ or its sub-classes. We write $\adom(Q)$ for the set of constants occurring in a query $Q$, and for a database instance $\bfI$, we also set $\adom(Q,\bfI) = \adom(Q) \cup \adom(\bfI)$.

In Section \ref{sec:combined}, we provide a characterization of which subclasses of $\UCQneg$ (i.e., $\SPJUneg$), defined by operations of the relational algebra, admit polynomial-time solvability of the missing answer problem. Rather than introduce the notation of the relational algebra, we instead look at the equivalent notions in our rule-based framework. We say that a semi-positive datalog program $P$ is \emph{self-join-free} if, for each rule $r$ in $P$, the body of $r$ does not contain more than one occurrence of a given relation symbol in $\Sigma$. We say that a semi-positive datalog program $P$ is \emph{projection-free} if, for each rule $r$ in $P$, we have that $\bound(r) = \emptyset$. We say that a semi-positive datalog program is \emph{selection-free} if it contains no occurrences of constants or equality atoms and none of its atoms contain duplicate variables. We say that $r$ is \emph{join-free} if the body of the rule contains only one atom. We write $\PJ$ for the selection-free sub-class of $\CQ$, $\SPUneg$ for the join-free sub-class of $\UCQneg$, and $\SJUneg$ for the projection-free sub-class of $\UCQneg$. We study these classes in Section \ref{sec:combined}.

As suggested in the introduction, the missing answer repair problem for a class $\logic$ of queries may be seen as a generalization of satisfiability problem for $\logic$, formally defined below.

\begin{definition}
\label{def:db-sat}
Let $\logic$ denote an arbitrary class of queries. The $\dbSATL$ problem is to determine, given a query $Q \in \logic$, whether or not there exists an instance $\bfI$ such that $Q(\bfI) \neq \emptyset$.
\end{definition}

It is well-known that this problem is undecidable for the relational calculus \cite{Trakhtenbrot50}, and hence also for relational algebra and stratified datalog. In addition to satisfiability, we will also compare the evaluation problem to the missing answer repair problem.

\begin{definition}
Let $\logic$ denote an arbitrary class of queries. Given an input of the form $\langle Q, \bfI, \tup{a} \rangle$, where $Q \in \logic$, $\bfI$ is an instance, and $\tup{a} \in \Dom$, the \emph{evaluation problem} for $\logic$, denoted $\EvalL$, is to determine whether or not $\tup{a} \in Q(\bfI)$.
\end{definition}

We write $\EvalQ$ to denote the data complexity version of this problem, where the query $Q$ is fixed, and inputs are of the form $\langle \bfI, \tup{a} \rangle$. Finally, for ease of reference, we define the following notion of the negation of a query.

\begin{remark}
\label{remark:negation}
In the introduction, we defined the negation of a query as follows: given a query $Q$, let $\lnot Q$ (the \emph{negation} of $Q$) denote a query which, on each database instance $\bfI$, returns the complement of $Q(\bfI)$ (with respect to all tuples of appropriate arity over the active domain of $\bfI$). Note that we do not specify to which query class the negation of a query belongs; we define it only in terms of the map between database instances that it represents. However, it should be noted that if $Q$ is expressible in the relational algebra ($\RA$), then $\lnot Q$ is also expressible in $\RA$ as the difference of $D^\arity(Q)$ and $Q$, where $D^\arity(Q)$ is the $\UCQ$ query that returns the set of all tuples of length $\arity(Q)$ over the active domain of an instance.
\end{remark}

\subparagraph*{Complexity Definitions}
We assume familiarity with the usual classes $\P$, $\NP$, $\EXP$, $\NEXP$, $\twoEXP$, and $\twoNEXP$. We prepend an ``F'' to a class to indicate the corresponding class of function problems (e.g., $\FP$, $\FEXP$).

\begin{definition}
An approximation problem is a triple $P = (\mathcal{I},\mathcal{S},\cost)$, where $\mathcal{I}$ is a set of \emph{instances}, $\mathcal{S}$ is a map from instances $x$ to sets $\mathcal{S}(x)$ of \emph{feasible solutions}, and $\cost$ is a map from pairs $(x,S(x))$ to natural numbers. The \emph{optimal cost} for an instance $x \in \mathcal{I}$ is $\cost^\ast(x) := \min_{y \in \mathcal{S}(x)} \cost(x,y)$. We write $\opt(x)$ to denote the set of solutions such that $\cost(x) = \cost^\ast(x)$.
\end{definition}

\begin{definition}[Strict reduction, \cite{OrponenM87}]
\label{def:strict-reduction}
Consider approximation problems $P_1 = (\mathcal{I}_1,\mathcal{S}_1,\cost_1)$ and $P_2 = (\mathcal{I}_2,\mathcal{S}_2,\cost_2)$. A \emph{strict reduction} from $P_1$ to $P_2$ is a pair of polynomial-time computable functions $(f,h)$ where $f: \mathcal{I}_1 \to \mathcal{I}_2$ and $h$ is a map from pairs $(x,y)$, where $x \in \mathcal{I}$ and $y \in \mathcal{S}_2(f(x))$, to elements of $\mathcal{S}_1(x)$, such that
\begin{enumerate}
\item If $x \in \mathcal{I}_1$, $y \in \mathcal{S}_2(f(x))$, and $y \in \opt(f(x))$, then $h(x,y) \in \opt(x)$, and
\item for all $x \in \mathcal{I}_1$ and $y \in \mathcal{S}(f(x))$, we have that $\cost(h(x,y)) \leq \cost(y)$.
\end{enumerate}
\end{definition}

\begin{definition}
The class $\FP^\NP$ is the class of function problems computable in polynomial time with access to an $\NP$ oracle. For a time-constructible function $l(n)$, we write $\FP^\NP[l(n)]$ for the class of function problems computable in polynomial time with $O(l(n))$ queries to an $\NP$ oracle, where $n$ is the size of the input.
\end{definition}

\begin{definition}[$\NP$ metric Turing machines, \cite{Krentel88}]
An $\NP$ metric Turing machine $M$ is a nondeterministic polynomially-time-bounded Turing machine such that every branch writes a binary number and accepts or rejects. The output of such a machine on an input $x \in \Sigma^\ast$ is the smallest value on any accepting branch of $M$ on the input $x$, denoted $\opt^M(x)$.
\end{definition}

\begin{definition}[Optimization Polynomial Time, \cite{Krentel88}]
A function $f: \Sigma^\ast \to \mathbb{N}$ is in $\OptP$ if there is an $\NP$ metric Turing machine $M$ such that $f(x) = \opt^M(x)$ for all $x \in \Sigma^\ast$. Furthermore, $f$ is in $\OptP[l(n)]$ if $f \in \OptP$ and the binary encoding of $f(x)$ is bounded by $l(\card{x})$ for all $x \in \Sigma^\ast$.
\end{definition}

\begin{definition}[Metric reductions, \cite{Krentel88}]
\label{def:metric-reduction}
Let $P_1,P_2: \Sigma^\ast \to \mathbb{N}$. A \emph{metric reduction} from $P_1$ to $P_2$ is a pair of polynomial-time computable functions $(f,g)$ where $f: \Sigma^\ast \to \Sigma^\ast$ and $g: \Sigma^\ast \times \mathbb{N} \to \mathbb{N}$ such that $P_1(x) = g(x,P_2(f(x)))$ for all $x \in \Sigma^\ast$. A metric reduction $(f,g)$ is \emph{linear} if the map $k \mapsto g(x,k)$ is linear and \emph{exact} if $g(x,k)=k$.
\end{definition}

\section{Formalizing the Missing Answer Problem}
\label{sec:ma}
To address the problem of finding minimal repairs for missing answers, we also study the problems of determining if a repair exists at all, determining if a repair within a given size bound exists, and determining the size of a minimal repair. Thus we define repairs as updates placing the desired tuple in the query answer, without reference to the size of the update.

\begin{definition}[Updates]
\label{def:repairs}
An \emph{update} for a database instance $\bfI$ is a pair $\Rep = (\Ins,\Del)$, where $\Ins \subseteq \Facts(\Dom) \setminus \bfI$ and $\Del \subseteq \bfI$ are finite sets of facts to be \emph{inserted} to and \emph{deleted} from $\bfI$. Given an update $\Rep$ for $\bfI$, we define $\repI = \left( \bfI \cup \Ins \right) \setminus \Del$. The size of an update $\Rep$ is $\card{\Rep} = \card{\Ins \cup \Del}$. We say that $\Rep$ is an \emph{update over $D \subseteq \Dom$} if $\Ins \cup \Del \subseteq \Facts(D)$. We say that an update $\Rep$ is a \emph{repair} with respect to a triple $\langle Q, \bfI, \tup{a} \rangle$ if $\tup{a} \in Q(\repI)$.
\end{definition}

In other words, the size of an update is the number of insertions and deletions that it contains. Note that $\card{\Rep} = \card{\bfI \oplus \repI}$; i.e., a \emph{minimal} repair is one which minimizes the size of the symmetric difference between the input instance $\bfI$ and the updated instance $\repI$. This measure of minimality is similar to the notion of a repair in consistent query answering under the cardinality-based repair semantics \cite{arenas2003answer, lopatenko06complexity}.

\begin{definition}[The missing answer problems]
    \label{def:MA}
    Let $\logic$ denote an arbitrary class of queries. We define four computational problems.
    \begin{enumerate}
        \item $\MAdecL$ (the missing answer decision problem): given $\langle Q, \bfI, \tup{a} \rangle$, determine whether or not there exists a repair $\Rep$ with respect to $\langle Q, \bfI, \tup{a} \rangle$;
        \item $\MAsizeL$ (the missing bounded missing answer repair problem): given $\langle Q, \bfI, \tup{a}, k \rangle$, determine whether or not there exists a repair $\Rep$ with respect to $\langle Q, \bfI, \tup{a} \rangle$ such that $\card{Rep} \leq k$;
        \item $\MAsizeL$ (the missing answer repair size problem): given an input $\langle Q, \bfI, \tup{a} \rangle$ compute the minimum size of a repair $\Rep$ with respect to $\langle Q, \bfI, \tup{a} \rangle$, or to otherwise indicate that a repair does not exist; and
        \item $\MAminL$ (the missing answer repair problem): given $\langle Q, \bfI, \tup{a} \rangle$, compute a minimum-cardinality repair $\Rep$ with respect to $\langle Q, \bfI, \tup{a} \rangle$, or otherwise indicate that a repair does not exist;
    \end{enumerate}
    where $Q \in \logic$, $\bfI$ is an instance, $\tup{a}$ is a tuple over $\Dom$ with $\tup{a} = \arity(Q)$, and $k \in \mathbb{N}$.
\end{definition}

Note that the above definitions allow arbitrary queries of the class $\logic$ to appear in the input -- this is the \emph{combined complexity} setting, discussed in Section \ref{sec:combined}. We will also study $\MAboundQ$, the variant of the bounded missing answer repair problem  in which the query $Q$ is fixed, inputs are of the form $\langle \bfI, \tup{a}, k \rangle$, and the goal is to determine whether or not a repair $\Rep$ with respect to $\langle Q, \bfI, \tup{a} \rangle$ such that $\card{\Rep} \leq k$ exists. Similarly, we study $\MAsizeQ$ and $\MAminQ$, the missing answer repair size and missing answer repair problems in which the query $Q$ is fixed, inputs are of the form $\langle \bfI, \tup{a} \rangle$, and the goal is to compute a repair $\Rep$ with respect to $\langle Q, \bfI, \tup{a} \rangle$ of minimal cardinality. The $\MAboundQ$, $\MAsizeQ$, and $\MAminQ$ problems are the \emph{data complexity} versions, which are the focus of Section \ref{sec:data}. In both the data and combined complexity settings, we consider the schema $\bfR$ of extensional database predicates to be fixed.

\subparagraph*{Preliminary complexity observations.}
Before turning to our main results, we begin with some preliminary observations.

\begin{proposition}
\label{prop:eval-to-rep}
Let $\logic$ be a class of queries. Then
\begin{enumerate}
\item $\EvalL \leq_\P \MAboundL$; and
\item $\EvalL$ can be decided in linear time with a single call to an $\MAsizeL$ (or $\MAminL$) oracle.
\end{enumerate}
\end{proposition}
\begin{proof}
The first claim is by the reduction $\langle Q, \bfI, \tup{a} \rangle \mapsto \langle Q, \bfI, \tup{a}, 0 \rangle$. The second claim follows from the observation that $\MAsizeL(Q,\bfI,\tup{a}) = 0$ if and only if $\MAminL(Q,\bfI,\tup{a}) = (\emptyset,\emptyset)$ if and only if $\EvalL(Q,\bfI,\tup{a})$ returns true. The linear running time is only required in order to copy the input to the oracle tape, call the oracle, and return the answer.
\end{proof}

Note that the relationships in Proposition \ref{prop:eval-to-rep} also hold for the data complexity versions of the problems. Since $\Eval{\CQ}$ is $\NP$-complete \cite{Chandra1977optimal}, it follows immediately from Proposition \ref{prop:eval-to-rep} that $\MAbound{\CQ}$ is $\NP$-hard, and that $\MAmin{\CQ}$ cannot be solved in polynomial time (unless $\P = \NP$). We refine these statements significantly in Section \ref{sec:combined}. Proposition \ref{prop:eval-to-rep} also implies that $\MAbound{\DL}$ is $\EXP$-hard, since $\Eval{\DL}$ is $\EXP$-complete~\cite{DantsinEGV97} (implicit in~\cite{Vardi82,Immerman86}).

We now consider some basic relationships between the different versions of the missing answer problem. Clearly, the $\MAminL$ problem is at least as hard as $\MAdecL$, $\MAboundL$, and $\MAsizeL$; similarly, $\MAminQ$ is at least as hard as $\MAsizeQ$ and $\MAboundQ$. Furthermore, $\MAboundL$ is the natural ``decision version'' of the $\MAminL$ optimization problem. There is no immediate relationship between $\MAboundL$ and $\MAdecL$, since for a given triple $\langle Q, \bfI, \tup{a} \rangle$, there may be a repair, but not one of size $k$. However, a straightforward relationship can be identified for classes of queries with the following property.

\begin{definition}[Computable bounded repairs]
\label{def:bounded-repairs}
We say that a class of queries $\logic$ \emph{admits computable bounded repairs} if there exists a computable function $\beta_\logic$ which takes as input tuples of the form $\langle Q, \bfI, \tup{a} \rangle$ and returns a finite set $S \subseteq \Dom$ such that, if a repair $\Rep$ with respect to $\langle Q, \bfI, \tup{a} \rangle$ exists, then a repair $\Rep'$ with respect to $\langle Q, \bfI, \tup{a} \rangle$ over $S$ also exists.
\end{definition}

The intuition for Definition \ref{def:bounded-repairs} is in the name: for a class of queries $\logic$ which admits computable bounded repairs, there is a computable function which outputs a finite search space of repairs for any missing answer input for that class. A notion similar to Definition \ref{def:bounded-repairs} will appear in Section \ref{sec:data} when analyzing the data complexity of the $\MAminQ$ problem for $\spDL$ queries $Q$ (cf. Definition \ref{def:constant-sized-repairs}). It is straightforward from Definition \ref{def:bounded-repairs} to verify the following proposition.

\begin{proposition}
\label{prop:dec-to-bound}
If $\logic$ is a class of queries which admits computable bounded repairs and $\EvalL$ is decidable, then the map $\langle Q, \bfI, \tup{a} \rangle \to \langle Q, \bfI, \tup{a}, p_\Sigma(\card{\boundL(Q, \bfI, \tup{a})}) \rangle$ is a computable reduction from $\MAdecL$ to $\MAboundL$, where $p_\Sigma(n) = \sum_{R \in \Sigma} n^{\arity(R)}$.
\end{proposition}

Importantly, although the $\MAboundL$ problem for a class $\logic$ of queries appears to have the flavor of an $\NP$ problem, this is not true in general. As an example, the combined complexity of $\DL$ evaluation is $\EXP$-complete~\cite{DantsinEGV97,Vardi82,Immerman86}, and so $\MAboundL$ is not in $\NP$ (unless $\NP = \EXP$), since $\Eval{\DL} \leq_P \MAbound{\DL}$ per Proposition \ref{prop:eval-to-rep}. However, since inputs to $\MAboundL$ include the integer $k$ which bounds the total number of insertions and deletions in witness repairs, it suffices to consider repairs which are at most \emph{exponentially-large} with respect to the input. This allows us to state a general upper bound, which requires only that the class of queries in question has an evaluation problem in $\NEXP$, and that the class contains only \emph{generic} queries, as defined below.

\begin{definition}[Generic queries]
\label{def:generic-queries}
Let $\rho: \Dom \to \Dom$ be an arbitrary bijective map. Given an instance $\bfI$, we write $\rho(\bfI)$ to denote the instance obtained by replacing each constant in each fact of $\bfI$ with its image under $\rho$. We define $\rho(\Rep)$ similarly for an update $\Rep$. A query $Q$ is $D$-generic for a subset $D \subseteq \Dom$ if, for any database $\bfI$ and any bijective map $\rho: \Dom \to \Dom$ which fixes $D$ point-wise (i.e., such that $\rho(d) = d$ for all $d \in D$), we have that $\rho(Q(\bfI)) = Q(\rho(\bfI))$. We say that a query $Q$ is \emph{generic} if it is $\adom(Q)$-generic.
\end{definition}

In other words, a query is generic if its answer is invariant to renaming constants not explicitly appearing in the active domain of the query. We are now ready to state and prove the general upper bound for $\MAboundL$.

\begin{proposition}
\label{prop:2NEXP-generic-MAbound}
Let $\logic$ be a class of generic queries. If $\EvalL$ is in $\NEXP$, then $\MAboundL$ is in $\twoNEXP$.
\end{proposition}
\begin{proof}[Proof (sketch).]
By genericity, for an input $\langle Q, \bfI, \tup{a}, k \rangle$ of size $N$, it suffices to consider repairs over $\tup{a} \cup \adom(Q,\bfI) \cup C$, where $C$ is a domain of fresh constants of size $k$. The algorithm proceeds as follows: non-deterministically generate an update $\Rep$ of cardinality at most $k$, and then non-deterministically generate a possible witness for $\EvalL$ for the instance $\repI$. The algorithm concludes by computing $\EvalL(\repI)$ and returning the answer. This algorithm is clearly correct. Furthermore, note that, since $k$ is represented in binary, the instance $\repI$ may be exponentially-larger than the encoding of the input. Then  because the witness for $\EvalL$ may be exponentially-larger than $\repI$, the algorithm is in $\twoNEXP$.
\end{proof}

Since $\Eval{\spDL}$ is in $\EXP$ and $\Eval{\RA}$ is in $\PSPACE$~(\cite{abiteboul1995foundations}), the above result implies that $\MAbound{\DL}$, $\MAbound{\spDL}$, and $\MAbound{\RA}$ are in $\twoNEXP$. Note that the argument for Proposition \ref{prop:2NEXP-generic-MAbound} also yields a $\twoNEXP$ upper bound for the complexity of $\MAboundL$ even if we assume that $\EvalL$ is in $\EXP$; we state the proposition for the hypothesis that $\EvalL$ is in $\NEXP$ to make it as general as possible. In Section \ref{sec:combined}, we refine the upper bound for $\MAbound{\DL}$ (resp. $\MAbound{\spDL}$) down to $\EXP$ (resp. $\twoEXP$) by giving a $\FEXP$ (resp. $\FtwoEXP$) algorithm for $\MAmin{\DL}$ (resp. $\MAmin{\spDL}$).

While $\MAboundL$ is not always an $\NP$ problem, we can give a very simple general criterion for $\MAboundL$ to be in $\NP$. Say that a class of queries $\logic$ \emph{admits polynomially-bounded repairs} if $\logic$ admits computable bounded repairs and $\beta_\logic$ is computable in polynomial time. Then by a similar algorithm to the proof of Proposition \ref{prop:2NEXP-generic-MAbound}, we obtain the following proposition.

\begin{proposition}
\label{prop:general-combined-UB}
Let $\logic$ be a class of generic queries. If $\logic$ admits polynomially-bounded repairs and $\EvalL$ is in $\NP$, then $\MAboundL$ is in $\NP$.
\end{proposition}

Recall that a query $Q$ in $\UCQneg$ is satisfied in an instance if there exists a satisfying assignment from the domain of the instance to one of the conjunctive queries in the definition of $Q$. Let $m$ denote the number of variables in the conjunctive query in the definition of $Q$ with the most variables. The range of a satisfying assignment for $Q$ has size at most $m$, and so it suffices to consider repairs over the domain $\tup{a} \cup \adom(Q,\bfI) \cup C$, where $C$ is a set of $m$ fresh constants. We conclude that $\UCQneg$ admits polynomially-bounded repairs. Furthermore, it is already known that $\Eval{\UCQneg}$ is in $\NP$~\cite{abiteboul1995foundations}. Hence $\UCQneg$ satisfies the hypotheses of Proposition \ref{prop:general-combined-UB}, and so the $\MAbound{\UCQneg}$ problem (and $\MAboundL$ for all syntactic fragments $\logic$ of $\UCQneg$) are in $\NP$.

\section{Combined Complexity}
\label{sec:combined}
We now show that for many classes of queries $\logic$, the $\MAdecL$ (see Definition \ref{def:MA}) and $\dbSATL$ (see Definition \ref{def:db-sat}) problems are mutually polynomial-time reducible, and we deduce a number of complexity results as a corollary. We then show that $\MAmin{\UCQsjf}$, $\MAmin{\SJUneg}$, and $\MAmin{\SPUneg}$ can be solved in polynomial time. From this, it also follows that $\MAbound{\UCQsjf}$, $\MAbound{\SJUneg}$, and $\MAbound{\SPUneg}$ are in $\P$.  Recall that $\SJUneg$ and $\SPUneg$ are the projection-free and join-free fragments of $\UCQneg$, respectively. Then, we provide a strict polynomial-time reduction, which doubles as a metric reduction (see Section~\ref{sec:prelim} for the reduction definitions), which shows that $\MAbound{\PJ}$ is $\NP$-complete and $\MAsize{\PJ}$ is $\OptP[\log(n)]$-hard. We conclude by showing that there exist an $\OptP[\log(n)]$ algorithm for $\MAsize{\PJ}$, an $\FP^{\NP}$ algorithm for $\MAmin{\UCQneg}$, and a $\FEXP$ algorithm for $\MAminL$.

\subparagraph*{The missing answer decision problem.}
\label{subsec:madec}
For our results relating $\MAdecL$ and $\dbSATL$, we need only to assume that the query class $\logic$ is closed under a weak form of selection and projection. We say that a class $\logic$ of queries is \emph{closed under (weak) selection} if there is a polynomial-time computable map which, given an $n$-ary query $Q \in \logic$ and a tuple $\tup{a}$ of length $n$, constructs a query $Q' \in \logic$ such that for all instances $\bfI$, we have $Q'(\bfI) \neq \emptyset$ if and only if $\tup{a} \in Q(\bfI)$. We say that a class $\logic$ of queries is \emph{closed under (weak) projection} if there is a polynomial-time computable map which, given $Q \in \logic$, constructs a query $Q' \in \logic$ such that for all instances $\bfI$ we have $Q'(\bfI)=\top$ if and only if $Q(\bfI) \neq \emptyset$.

\begin{theorem}
\label{thm:MAdec-to-SAT}
If $\logic$ is closed under (weak) selection, then $\MAdecL \Pred \dbSATL$.
\end{theorem}
\begin{proof}
Given an input $\langle Q, \bfI, \tup{a}\rangle$ to $\MAdecL$, let $Q'$ be the query in $\logic$ guaranteed by closure under weak selection. If $\langle Q, \bfI, \tup{a} \rangle \in \MAdecL$, then $Q'(\repI) \neq \emptyset$, and so $Q' \in \dbSATL$. Conversely, if $Q'\in \dbSATL$, then there exists an instance $\bfJ$ such that $Q'(\bfJ) = \top$, and so $\tup{a} \in Q(\bfJ)$. Then for the repair $\Rep = (\Ins, \Del)$ where $\Ins = \bfJ \setminus \bfI$ and $\Del = \bfI \setminus \bfJ$, we have that $\repI = \bfJ$. Hence $\tup{a} \in Q(\repI)$, and so we conclude that $\langle Q, \bfI, \tup{a} \rangle \in \MAdecL$.
\end{proof}

The proof of the next theorem is similar to the proof of Theorem~\ref{thm:MAdec-to-SAT}.

\begin{theorem}
\label{thm:SAT-to-MAdec}
If $\logic$ is closed under (weak) projection, then $\dbSATL \Pred \MAdecL$.
\end{theorem}
\begin{proof}
Given an input $Q \in \logic$, let $\langle Q', \bfE, \epsilon \rangle$ be the $\MAdecL$ input where $Q'$ is the query in $\logic$ guaranteed by closure under weak projection, $\bfE$ is the empty instance, and $\epsilon$ is the empty tuple. If $Q \in \dbSATL$, then $Q(\bfJ) \neq \emptyset$ for some $\bfJ$, and so $Q'(\bfJ) = \top$. Hence the repair $\Rep = (\Ins,\Del)$ with $\Ins = \bfJ$ and $\Del = \emptyset$ yields a repaired instance $\repI = \bfJ$. Thus $\langle Q', \bfE, \epsilon \rangle \in \MAdecL$. Conversely, if $\langle Q', \bfE, \epsilon \rangle \in \MAdecL$, then there exists a repair $\Rep$ such that $Q'(\bfE_\Rep) = \top$. Hence $Q(\bfE_\Rep) \neq \emptyset$, and so $Q \in \dbSATL$.
\end{proof}

All query classes mentioned in Section \ref{sec:prelim} satisfy the weak selection and projection properties, and the maps are in fact computable in time linear in the size of the query.

\begin{corollary}
\label{cor:MAdec}
The following statements hold.
\begin{enumerate}
\item $\MAdec{\UCQneg}$ is decidable in $\P$;
\item $\MAdec{\DL}$ and $\MAdec{\spDL}$ are $\EXP$-complete;
\item $\MAdec{\RA}$ is undecidable.
\end{enumerate}
\end{corollary}
\begin{proof}
By Theorems \ref{thm:MAdec-to-SAT} and \ref{thm:SAT-to-MAdec}, together with known results on the complexity of the satisfiability problems for these classes \cite{abiteboul1995foundations, Shmueli93, levy1993equivalence, Trakhtenbrot50}. Note that, to the author's knowledge, the $\EXP$-hardness of $\dbSAT{\DL}$ is an unpublished folklore result, and so we include a proof of this lower bound in the appendix.
\end{proof}

Since $\MAdecL$ can be solved with an oracle for $\MAsizeL$ or $\MAminL$, Corollary \ref{cor:MAdec} implies $\MAsize{\DL}$, $\MAsize{\spDL}$, $\MAmin{\DL}$, and $\MAmin{\spDL}$ are $\EXP$-hard, while $\MAsize{\RA}$ and $\MAmin{\RA}$ are uncomputable.

\subparagraph*{The (bounded) missing answer repair problem.} We now turn to the complexity of the $\MAboundL$ and $\MAminL$ problems for sub-classes of $\UCQneg$. For the polynomial-time cases of $\MAboundL$, it clearly suffices to show that $\MAminL$ can be solved in polynomial time.

\begin{theorem}
\label{thm:combined-poly-classes}
The $\MAmin{\UCQsjf}$, $\MAmin{\SJUneg}$, and $\MAmin{\SPUneg}$ problems are in $\FP$.
\end{theorem}
\begin{proof}
To find a minimal repair for a query $Q \in \UCQ$ which is the union of conjunctive queries $Q_1,\hdots,Q_m$, it suffices to find minimal repairs $\Rep_i$ for each $Q_i$ and return the smallest of these. Thus we need only to show that minimal repairs for $\CQsjf$, $\SJneg$, and $\SPneg$ queries can be found in polynomial time. For the first, observe that for a self-join-free conjunctive query, the body contains at most $\card{\Sigma}$ literals. Since the schema is fixed, we can form a repair for a triple $\langle Q, \bfI, \tup{a} \rangle$ with only a constant number of insertions or deletions over the domain $\tup{a} \cup \adom(Q,\bfI) \cup \{ c \}$, where $c$ is a fresh constant. Furthermore, the $\Eval{\CQsjf}$ is in polynomial-time, and so we can solve $\MAmin{\CQsjf}$ in polynomial-time with a brute force search of all possible repairs.

Now suppose we are given an input $\langle Q, \bfI, \tup{a} \rangle$ to $\MAmin{\SJneg}$. Recall that $Q$ is defined by a rule $r$ of the form $S(\tup{x}) \df \alpha_1, \hdots, \alpha_{n}$, where each $\alpha_i$ a positive or negative literal and $\{ \tup{x} \} = \bigcup_{i \leq n} \var(\alpha_i)$. Thus $\Body_r(\tup{a})$ is an expression containing no variables, and so for any instance $\bfJ$, we have that $\tup{a} \in Q(\bfJ)$ if and only if $\alpha_i[g] \in \bfJ$ for each $i \leq n$, where $g$ is the assignment with $g(\tup{x}) = \tup{a}$. Hence the minimal repair is given by $\Rep = (\Ins, \Del)$ where
\begin{align*}
\Ins &= \{ R(\tup{a}) \mid \alpha_i = R(\tup{x}) ~\text{and}~ R(\tup{a}) \not \in \bfI ~\text{for some}~ i \leq n \} \\
\Del &= \{ R(\tup{a}) \mid \alpha_i = \lnot R(\tup{x}) ~\text{and}~ R(\tup{a}) \in \bfI ~\text{for some}~ i \leq n \},
\end{align*}
which can be easily computed in polynomial-time.

Now suppose $\langle Q, \bfI, \tup{a} \rangle$ is an input to $\MAmin{\SPneg}$. Recall that $Q$ is defined by a rule of the form $S'(\tup{x}) \df \beta$, where $\beta$ is a positive or negative literal which may contain variables not occurring in the tuple $\tup{x}$. Clearly, either $\tup{a} \in Q(\bfI)$, or the minimal repair contains exactly one insertion (if $\beta$ is a positive literal) or exactly one deletion (if $\beta$ is a negative literal). In either case, it is straightforward to iterate through all facts in the appropriate relation of the database to determine which case is required, and to compute a minimal repair. Since this is clearly a polynomial-time procedure, we're done.
\end{proof}

\begin{theorem}
\label{thm:combined-UCQneg-upper-bounds}
$\MAsize{\UCQneg}$ is in $\OptP[O(\log(n))]$ and $\MAmin{\UCQneg}$ is in $\FP^\NP$.
\end{theorem}
\begin{proof}
We know that $\UCQneg$ has polynomially-bounded repairs (see Section \ref{sec:ma}). Hence we can solve $\MAsize{\UCQ}$ with an $\NP$ metric Turing machine as follows. Given an input $\langle Q, \bfI, \tup{a} \rangle$, on each branch of the computation, we non-deterministically generate a possible update, non-deterministically generate a possible satisfying assignment, and then check if the satisfying assignment witnesses that the update is a repair with respect to $\langle Q, \bfI, \tup{a} \rangle$. If yes, then we write down the size of the repair in binary and accept. Otherwise, we reject on that branch. This algorithm is clearly correct, and each branch runs in polynomial time.

For membership of $\MAmin{\UCQneg}$ in $\FP^{\NP[O(n^2)]}$, let $\langle Q, \bfI, \tup{a} \rangle$ be an input with $Q$ in $\UCQneg$. We first determine the size $K$ of the optimal repair (if it exists) by a binary search, using $\NP$ queries to $\MAbound{\UCQneg}$, on the interval $[0,m]$, where $m$ is the number of literals in the rule of $Q$ with the most literals. Then, we construct an optimal repair as follows. First, we let $C$ denote a set of constants distinct from $\tup{a} \cup \adom(\bfI)$ equal to the number of variables occurring in $Q$, and we initialize a partial assignment $g = \{ \langle x_i,a_i \rangle \mid i \leq \len(\tup{a}) \}$. For each variable $x$ in $Q$ and each element $b$ in $\tup{a} \cup \adom(D) \cup C$, we query the oracle for an answer to ``does there exist a repair with respect to $\langle Q, \bfI, \tup{a} \rangle$ of size $K$ where the fact that $\tup{a} \in Q(\bfI)$ is witnessed by an assignment extending $g \cup \{\langle x,b \rangle \}$?'' It is straightforward to see that this query is in $\NP$. Furthermore, we need only ask $O(m(\card{\tup{a}} + \card{\adom(\bfI)}+m))$ queries to this oracle, which is quadratic in the size of the input. Hence $\MAmin{\UCQneg}$ is in $\FP^{\NP[O(n^2)]}$.
\end{proof}

We now provide $\OptP[\log(n)]$ and $\FP^{\NP[\log(n)]}$ lower bounds for $\MAbound{\PJ}$ and $\MAmin{\PJ}$, respectively, by providing a strict and metric reduction to the minimum set cover\footnote{The $\minSC$ problem is, given a collection $X = \{ s_1,\hdots,s_m \}$ of sets, where $U = \bigcup_{i \leq m} s_i$, to find a set $C \subseteq X$ such that for each $u \in U$, there exists $s \in C$ such that $u \in s$.} ($\minSC$) problem. See Section \ref{sec:prelim} for the definitions of strict and metric reductions. Note that we view $\MAminL$ as an approximation algorithm by taking repairs $\Rep$ with respect to $\langle Q, \bfI, \tup{a} \rangle$ as \emph{feasible solutions} for \emph{instances} $\langle Q, \bfI, \tup{a} \rangle$, and the size $\card{\Rep}$ of the repair as its \emph{cost}. It is known that if there exists an approximation algorithm for $\minSC$ with approximation ratio $(1 - \alpha) \ln(n)$ for any $\alpha > 0$, where $n$ is the size of the instance, then $\P = \NP$ \cite{DinurS14}.

\begin{theorem}
\label{thm:PJ-hard}
If the (EDB) schema $\bfR$ contains a unary and a binary relation symbol, then
\begin{enumerate}
\item the $\MAbound{\PJ}$ problem is $\NP$-complete, and
\item the $\MAsize{\PJ}$ problem is $\OptP[\log(n)]$-complete.
\item the $\MAmin{\PJ}$ problem is $\FP^{\NP[\log(n)]}$-hard.
\item the $\MAmin{\PJ}$ problem is $\minSC$-hard to approximate under strict reductions.
\end{enumerate}
\end{theorem}
\begin{proof}
The $\NP$-hardness result follows from $\NP$-hardness of $\Eval{\PJ}$ and the reduction given in Proposition \ref{prop:eval-to-rep}. For the approximation hardness result, we provide a strict reduction $(f,h)$ from the $\minSC$ problem to the $\MAmin{\PJ}$ problem. This reduction is also a metric reduction from the $\minSC$ problem to the $\MAmin{\PJ}$ problem, from which it follows that $\MAsize{\PJ}$ problem is $\OptP[\log(n)]$-hard. Furthermore, the $\FP^{\NP[\log(n)]}$-hardness of $\MAmin{\PJ}$ follows from the $\OptP[\log(n)]$-hardness of $\MAsize{\PJ}$ by~\cite[Theorem 3.3]{Krentel88}.

Let $X = \{ s_1, \hdots, s_m \}$ with $U = \bigcup_{i \leq m} s_i = \{ u_1, \hdots, u_n \}$ denote an arbitrary input to $\minSC$. Let $P$ be a unary relation symbol in $\bfR$ and $F$ be a binary relation symbol in $\bfR$, and fix distinct constants $\tup{a} = a_1,\hdots,a_n$ and $\tup{b} = b_1,\hdots,b_m$. We now define an input $f(X) = \langle Q^X, \bfI^X, a_1,\hdots,a_n \rangle$ to $\MAmin{\PJ}$ as follows. Let $Q^X$ denote the $\PJ$ query defined by the rule
\[
\ans(x_1,\hdots,x_n) \df F(y_1,x_1), \hdots, F(y_n,x_n), P(y_1), \hdots, P(y_n),
\]
where $x_1,\hdots,x_n,y_1,\hdots,y_n$ are pairwise distinct variables. Set $\bfI^X = \{ F(b_j,a_i) \mid u_i \in s_j \}$. The idea is that the $a_i$ constants represent elements of $U$, the $b_i$ constants represent elements of $S$, and a fact $F(a,b)$ asserts that $a \in b$. The unary $P$ predicate is used to ``choose'' which sets to include in the set cover. We now make this intuition precise.

We claim that $\MAmin{\PJ}(f(X)) = k$. To see that $\MAmin{\PJ}(f(X)) \leq k$, let $C$ be a set cover of $U$ of cardinality $k$. Then for all $i \leq n$, we have by construction of $\bfI^X$ that there exists some $s_j \in C$ such that $F(b_j,a_i) \in \bfI^X$. It follows easily that $\tup{a} \in Q(\bfI^X_\Rep)$ for the repair $\Rep = (\{ P(b_j) \mid s_j \in C \}, \emptyset)$. Hence $\MAmin{\PJ}(f(X)) \leq k$. To show that $\MAmin{\PJ}(f(X)) \geq k$, we use the following claim.

\begin{claim}
\label{claim:SC-reduction-claim}
If $\Rep = (\Ins,\Del)$ is a repair of least cardinality such that $\tup{a} \in Q'(\bfI^X_\Rep)$, then $\Del = \emptyset$. Furthermore, if $\Ins$ contains a fact $F(c,d)$ for some constants $c,d \in \Dom$, then there exists some $s_i \in S$ such that $\tup{a} \in Q(\bfI^X_{\Rep'})$, where $\Rep' = (\Ins',\Del)$ with $\Ins' = (\Ins \setminus \{ F(c,d) \}) \cup \{ P(b_i) \}$.
\end{claim}
\begin{proof}[Proof of claim.]
Let $\Rep = (\Ins,\Del)$ be as in the statement of the claim. That $\Del = \emptyset$ follows immediately from the monotonicity of $Q$. Since $\tup{a} \in Q(\bfI^X_\Rep)$, we must have that for each $i \leq n$, there exists $e \in \Dom$ such that $F(e,a_i), P(e) \in \bfI^X$. Now observe that $d = a_i$ for some $i \leq n$, because otherwise we have that $\tup{a} \in Q(\bfI^X_{\Rep''})$ where $\Rep'' = (\Ins \setminus \{ F(c,d) \}, \emptyset)$, contradicting the minimality of $\Rep$. Let $s_j \in S$ be an arbitrary set such that $u_i \in s_j$ (which must exist, since $\bigcup S = U$, by assumption) and set $\Rep' = (\Rep \setminus \{ F(c,a_i) \}) \cup \{ P(b_j) \}$. Since $\Del = \emptyset$, we have that $F(b_j,a_i), P(b_j) \in \bfI^x_{\Rep'}$, and so $\tup{a} \in Q(\bfI^x_{\Rep'})$, which is what we wanted to show.
\end{proof}

By Claim \ref{claim:SC-reduction-claim}, there is a repair $\Rep = (\Ins, \Del)$ of minimal cardinality such that $\tup{a} \in Q(\bfI^x_{\Rep})$, where $\Del = \emptyset$ and $\Ins$ contains only facts of the form $P(b_i)$ for sets $s_i \in S$. It follows that $C = \{ s_i \mid P(b_i) \in \Ins \}$ is a set cover of $U$, and so $\MAmin{\PJ}(f(X)) \geq k$.

For the map $h$, given an input $X$ to $\minSC$ and some repair $\Rep$ such that $a_1,\hdots,a_n \in Q^X(\bfI^X_{\Rep})$, we mimic the proof of Claim \ref{claim:SC-reduction-claim} to find a set cover of $U$ of better or equal cost. First, we discard all deletions and all insertions of facts of the form $F(c,d)$ where $d \neq a_i$ for some $i \leq n$. Then, we replace all insertions of $F$ facts of the form $F(c,a_i)$ for some $i \leq n$ with an arbitrary $P(b_j)$ fact where $u_i \in s_j$. The output of $h$ is the collection of sets corresponding to the remaining $P$ facts, which is clearly a set cover of $U$ of less or equal cost to $\Rep$. Note that the pair $(f,h)$ is both a strict reduction (see Definition \ref{def:strict-reduction}) and a metric reduction (see Definition \ref{def:metric-reduction}).
\end{proof}

We conclude this section with an upper bound for $\MAmin{\DL}$.

\begin{theorem}
\label{thm:DL-upper-bound}
$\MAmin{\DL}$ is in $\FEXP$ and $\MAmin{\spDL}$ is in $\FtwoEXP$.
\end{theorem}
\begin{proof}
For membership of $\MAmin{\DL}$ in $\FEXP$, recall that any Boolean $\DL$ program $P$ is monotone. Furthermore, it has been observed in earlier work (\cite{Shmueli93}) that $P$ can be satisfied in a database instance containing only the constants in $\adom(P)$. Given an input $\langle P, \bfI, \tup{a} \rangle$, we first form a Boolean $\DL$ program $P'$ by substituting the constants in $\tup{a}$ for the free variables of rules in $P$ with answer predicate $\ans_P$. By monotonicity, it suffices to only make insertions to $\bfI$ in order to repair $P'$. Furthermore, we can make only insertions over the set $\adom(P',\bfI)$ of constants, which is a subset of $\Dom$ which is linear in the size of the input. Consider the algorithm which, for each possible update $\Rep$ over $\Facts(\adom(P',\bfI)) \setminus \Facts(\adom(\bfI))$ (in order of cardinality of $\Rep$, checks if $\Rep$ is a repair with respect to $\langle P, \bfI, \tup{a} \rangle$ by computing $\Eval{\DL}(\repI)$. If so, then $\Rep$ is a minimal repair, and so we return $\Rep$. If no repair is found, then we reject. Since the loop uses linearly-many iterations in the size of the input, and simulating $\Eval{\DL}$ takes exponential time, it is easy to see that this algorithm is in $\FEXP$.

The argument that $\MAmin{\spDL}$ is in $\FtwoEXP$ is similar. In this case, given an input $\langle P, \bfI, \tup{a} \rangle$, we need to consider all repairs, including both insertions and deletions, over $\tup{a} \cup \adom(P,\bfI) \cup C$, where $C$ is set of fresh constants which is exponentially-large in the size of $P$. The sufficiency of this search space follows by the $\EXP$ satisfiability procedure for $\spDL$ given in \cite{levy1993equivalence}, which can be easily modified to allow for the construction of a satisfying database instance $\bfJ$ of at most exponential size. We can assume without loss of generality that $\bfJ$ contains only facts over $\adom(P) \cup C$. As a result, the brute force algorithm must check doubly-exponentially-many possible updates $\Rep$ over $\tup{a} \cup \adom(P,\bfI) \cup C$, where $\repI$ is guaranteed to be at most exponentially-larger than the input. Consequently, the entire brute-force algorithm runs in doubly-exponential time.
\end{proof}

\section{Data Complexity}
\label{sec:data}
In this section, we show that $\MAmin{P}$ is polynomial-time computable for all $\spDL$ programs $P$. This will be a special case of a general result on polynomial-time computability, under data complexity, of minimal repairs for a large class of queries. For this, we first introduce the following necessary definition, which is the key observation leading to the polynomial-time computability of $\MAmin{\spDL}$.

\begin{definition}
\label{def:constant-sized-repairs}
A query $Q$ \emph{admits constant-sized repairs} if there is some $\cbound \in \mathbb{N}$ such that, for any instance $\bfI$ and tuple $\tup{a}$, if a repair with respect to $\langle Q, \bfI, \tup{a} \rangle$ exists, then there exists a repair $\Rep'$ with respect to $\langle Q, \bfI, \tup{a} \rangle$ such that $\card{\Rep'} \leq \cbound$.
\end{definition}

Note that Definition \ref{def:constant-sized-repairs} describes a property of queries, rather than query classes. It also only involves a finite bound on the size of repairs, but does not necessarily imply that a bounded domain of constants for repairs can be computed. However, such a domain can be computed if the map $Q \mapsto c_Q$ is computable and $Q$ is generic (see Definition \ref{def:generic-queries}) -- this will be implicit in the proof of Theorem \ref{thm:general-data-ptime}. However, when studying data complexity, it does not matter whether or not $Q \to c_Q$ is computable, only that $c_Q$ exists for the query $Q$ of interest. The following example sheds some light on the notion of constant-sized repairs.

\begin{example}
Let $Q \in \CQneg$ be defined by the rule $S(x,y,z) \df R(x,y), R(y,z), \lnot R(z,x)$. Given an arbitrary instance $\bfI$ and a tuple $\langle a_1,a_2,a_3 \rangle$, a repair $\Rep$ needs at most two insertions and one deletion in order to ensure that $\langle a_1,a_2,a_3 \rangle \in \repI$ -- namely, we must have that $R(a_1,a_2), R(a_2,a_3) \in \repI$ and $R(a_3,a_1) \not \in \repI$. Hence $Q$ admits constant-sized repairs with $c_Q = 3$. Note that admitting constant-sized repairs does not mean that a repair always exists: for example, if the input tuple is $\langle a_1, a_1, a_1 \rangle$, then there is no repair $\Rep$ such that $\langle a_1,a_1,a_1 \rangle \in Q(\repI)$, since this would require that the fact $R(a_1,a_1)$ is both in and not in $\repI$. 
\end{example}

\begin{lemma}
\label{lemma:cqneg-constant-repairs}
Every query $Q \in \CQneg$ admits constant-sized repairs.
\end{lemma}
\begin{proof}
Let $Q \in \CQneg$ be an arbitrary query defined by a rule $r$ of the form
\[
S(\tup{x}) \df \alpha_1,\hdots,a_{n_1},\lnot \beta_1,\hdots,\lnot \beta_{n_2},\gamma_1,\hdots,\gamma_{n_3},
\]
where the $\alpha_i$ are positive literals, the $\lnot \beta_i$ are negative literals, the $\gamma_i$ are either inequality or equality atoms. Let $\bfI$ denote an arbitrary instance and $\tup{a}$ denote an arbitrary tuple over $\Dom$. To ensure that $\tup{a} \in \repI$, we need to find a repair $\Rep$ such that for some assignment $g$ with $g(\tup{x}) = \tup{a}$, we have that $\alpha_i[g] \in \repI$ for each $i \leq n_1$, $\beta_i[g] \not \in \repI$ for each $i \leq n_2$, and the $\gamma_i$ equalities/inequalities are satisfied by $g$. For any assignment $g$ which satisfies the $\gamma_i$, the maximum number of insertions needed is $n_1$, and the maximum number of deletions needed in $n_2$. Therefore, the maximum size of any minimal repair, if one exists at all, is $n_1 + n_2$.
\end{proof}

\begin{theorem}
\label{thm:general-data-ptime}
If $Q$ is an $\adom(Q)$-generic query which admits constant-sized repairs and such that $\EvalQ$ is polynomial-time computable, then $\MAminQ$ is in $\FP$.
\end{theorem}
\begin{proof}
Let $Q$ be as in the statement of the claim, and let $\langle \bfI, \tup{a} \rangle$ be an input to $\MAminQ$. Let $m = \max\{ \arity(R) \mid R \in \bfR \}$ and let $C = \{ c_1,\hdots,c_{m \cdot c_q} \}$ be a set of constants disjoint from $\adom(Q,I) \cup \tup{a}$. We use the following claim.

\begin{claim}
\label{claim:small-space}
If there exists a minimal cardinality repair $\Rep$ such that $\tup{a} \in Q(\repI)$, then there exists a minimal cardinality repair $\Rep'$ over $C \cup \adom(Q,I) \cup \tup{a}$.
\end{claim}
\begin{proof}[Proof of claim.]
Let $D$ denote the set of constants occurring in $\Rep$ which do not occur in $\adom(Q,\bfI) \cup \tup{a}$. Since $\Rep$ is minimal and $Q$ admits constant-sized repairs, we have $\card{D} \leq m \cdot c_Q$. By the $\adom(Q)$-genericity of $Q$, we have that $\rho(Q(\repI)) = Q(\rho(\repI))$ for any bijective map $\rho: \Dom \to \Dom$ with $\Img_\rho[D] \subseteq C$ and which fixes $\adom(Q) \cup \tup{a}$ point-wise. In particular, $\Rep' = \rho(\Rep)$, obtained by replacing each constant in $\Rep$ by its image under $\rho$, is also a minimal cardinality repair over $C \cup \adom(Q,I) \cup \tup{a}$ such that $\tup{a} \in Q(\bfI_{\Rep'})$.
\end{proof}

To find a minimal update $\Rep$ such that $\tup{a} \in Q(I)$, it suffices by Claim \ref{claim:small-space} to check, for all repairs $\Rep$ over $\adom(Q,\bfI) \cup C \cup \tup{a}$  (in order of size) such that $\card{\Rep} \leq c_Q$, whether or not $\tup{a} \in Q(\repI)$ (by a call to $\EvalQ$), outputting the first repair found. Since $Q$ and $c_Q$ are both constants, while $\adom(Q,\bfI) \cup C \cup \tup{a}$ is linear in the size of the input, this procedure checks only $O(n^{c_Q})$ repairs. Furthermore, since $\EvalQ$ is in $\FP$, this algorithm is also in $\FP$.
\end{proof}

We now recall the notion of \emph{unfoldings}~\cite{chaudhuri1997equivalence}. Intuitively, an unfolding of an $\spDL$ program $P$ is a conjunctive query whose head predicate is $\ans_P$, whose body contains only EDB predicates, and which is obtained by repeatedly substituting the IDB atoms of some rule of $P$ with the body of rules in $P$ whose heads contain the relevant IDB predicate.

\begin{definition}[Unfoldings]
Let $P$ be an $\spDL$ program with rules $r_1,\hdots,r_m$ and answer predicate $\ans_P$. We define $P^\infty$ to be the smallest set containing each rule of $P$ with $\ans_P$ in the head and such that, whenever $P^\infty$ contains a rule $r$ of the form
\[
\ans_P(\tup{x}) \df \alpha_1,\hdots,\alpha_{i-1}, \alpha_i,\alpha_{i+1},\hdots,\alpha_n,
\]
where $\rel(\alpha_i) = S \in \bfS$, $\var(\alpha_i) = \tup{y}$, and $r'$ is a rule in $P$ of the form $S(\tup{u}) \df \Body_{r'}(\tup{u}, \tup{v})$ with $\free(r') = \tup{u}$ and $\bound(r') = \tup{v}$, then $P^\infty$ also contains the rule
\[
\ans_P(\tup{x}) \df \alpha_1,\hdots,\alpha_{i-1}, \Body_{r'}(\tup{y}, \tup{z}), \alpha_{i+1},\hdots,\alpha_n
\]
where $\tup{z}$ is a tuple of fresh variables not occurring in $r$ or $r'$. We define $\Unfoldings(P)$ to be the set of all rules $r$ in $P^\infty$ such that the body of $r$ contains only extensional relation symbols.
\end{definition}

The next theorem, originally stated for $\DL$ but generalizable to $\spDL$ programs, enables us to view $\spDL$ programs as (possibly infinite) unions of $\CQneg$ queries.

\begin{theorem}[\cite{chaudhuri1997equivalence}]
If $P \in \spDL$ has answer predicate $\ans(\tup{x})$, then
\[
\ans(\bfI) = \{ \tup{a} \in \Dom^{\len(\tup{x})} \mid \tup{a} \in Q(\bfI) ~\text{for some}~ Q \in \Unfoldings(P) \}.
\]
\end{theorem}

\begin{corollary}
\label{cor:spDL-main}
$\MAmin{P}$ can be computed in polynomial-time for all $P \in \spDL$.
\end{corollary}
\begin{proof}
Since it is known that evaluation for semi-positive datalog queries is computable in polynomial-time with respect to data complexity \cite{Immerman86}, it suffices by Theorem \ref{thm:general-data-ptime} to show that every $P \in \spDL$ admits constant-size repairs. For this, fix a program $P \in \spDL$, let $\adom(P) = a_1,\hdots,a_n$, and set $k = \arity(\ans)$, where $\ans$ is the answer predicate for $P$. Define an equivalence relation $\equiv$ on tuples in $\Dom^k$, where for tuples $\tup{b} = b_1,\hdots,b_k$ and $\tup{c} = c_1,\hdots,c_k$, we set $\tup{b} \equiv \tup{c}$ if and only if both of the following conditions hold:
\begin{enumerate}
\item $b_i = b_j$ if and only if $c_i = c_j$, and
\item $b_i = a_j$ if and only if $c_i = a_j$.
\end{enumerate}

If $\tup{b} \equiv \tup{c}$, then by $\adom(Q)$-genericity of $P$, we have that $\tup{b} \in \ans(\bfI)$ for some instance $\bfI$ if and only if $\tup{c} \in \ans(\rho(\bfI))$, where $\rho: \Dom \to \Dom$ is the bijective map obtained by setting $\rho(\tup{b}) = \tup{c}$,  $\rho(\tup{c}) = \tup{b}$, and $\rho(e) = e$ for all other $e \in \Dom$. Hence for each equivalence class $\mathcal{C}$ in $\Dom^k \setminus \equiv$, either there exists some $Q_\mathcal{C} \in \Unfoldings(P)$ such that, for all $\tup{b} \in \mathcal{C}$, there exists an instance $\bfI$ such that $\tup{b} \in Q_\mathcal{C}(\bfI)$, or no such $Q_\mathcal{C}$ exists. In the first case, set $N_\mathcal{C} = c_{Q_\mathcal{C}}$; otherwise, set $N_\mathcal{C} = 0$. Since there are finitely-many equivalences classes, we have that $c_P = \max_{\mathcal{C} \in \Dom^k \setminus \equiv} c_{Q_{\mathcal{C}}}$ is finite. Hence $P$ admits constant-sized repairs bounded by $c_P$, which is what we wanted to show.
\end{proof}

Since $\MAmin{P}$ is polynomial-time computable for programs $P$ in $\spDL$, it is natural to ask if the above argument can be extended to the relational algebra ($\RA$) or even to stratified datalog $(\DLneg)$. The following example shows that this is not the case.

\begin{example}
\label{ex:universal-query}
Consider the Boolean $\DLneg$ query, also expressible in $\RA$, given by the rules $R \df P(x)$ and $\ans \df \lnot R$, where $P$ is an EDB predicate and $R,\ans$ are IDB predicates. In the notation of the relational calculus, this is equivalent to the formula $\forall x \lnot P(x)$ (with the universal quantifier interpreted over the active domain). It is easy to see that, for any instance $\bfI$, we have $\ans(\bfI) \neq \emptyset$ if and only if $\bfI$ contains no facts of the form $P(a)$ for any $a \in \Dom$. Hence minimal repairs for this query may be arbitrarily large.
\end{example}

The above example illustrates an important theme: all queries $Q$ in $\spDL$ are \emph{existential}, and can be satisfied by modifying only a small \emph{local} part of the database. On the other hand, the query in Example \ref{ex:universal-query} is \emph{universal} in the sense that it expresses a \emph{global} condition on the active domain of the database. Consequently, minimal repairs for this query can be arbitrarily large, depending on the number of $P$ facts in the input database. While ``existential'' and ``local'' rules yield tractable complexity for the missing answer problems, the opposite is true for the wrong answer problems, by the duality discussed in the introduction.

Corollary \ref{cor:MAdec} provides a pessimistic picture of the combined complexity of determining whether or not repairs exist, particularly for $\RA$ and its extensions. However, in data complexity setting, the $\MAminQ$ problem is always decidable for $\RA$ queries. To see this, observe that, given a fixed query $Q \in \RA$, one of the following algorithms is correct:
\begin{enumerate}
\item if $Q$ is satisfiable: the algorithm which enumerates all possible repairs from smallest to largest until finding a repair; and
\item if $Q$ is unsatisfiable: the algorithm that returns ``no repair'' on all inputs.
\end{enumerate}
Of course, it is undecidable which of the above algorithms is correct, and even for $\RA$ queries which we know to be satisfiable, the running time of such a brute-force search is severely intractable. However, it seems that this is the best that we can say in general for $\RA$ queries, since there is no computable bound in general on the size of repairs required for such queries. If such a bound did exist, then $\dbSAT{\RA}$ would be solvable by checking all repairs within the computable bound for the input $\langle Q, \bfE, \epsilon \rangle$ -- which is impossible.

\section{Concluding Remarks and Future Work}
\label{sec:conc}
Our primary interest in this paper was to identify the dividing line between tractability and intractability of finding missing answer repairs. In combined complexity, we answer this question for various important sub-classes of $\UCQneg$ queries. We also provide $\EXP$-hardness of $\DL$ and $\spDL$, with a matching $\FEXP$ (resp. $\FtwoEXP)$) membership result for $\MAmin{\DL}$ (resp. $\MAmin{\spDL}$). In data complexity, we show that computing missing answer repairs is tractable for the expressive class of $\spDL$ queries. We also discuss the principal difficulty in extending tractability to queries with stronger forms of negation.

One avenue of future work is to close the gaps in Table \ref{table:combined}, such as those for $\MAsize{\spDL}$, $\MAmin{\spDL}$, and $\MAbound{\RA}$. As mentioned in the introduction, the missing answer repair also provides an elegant generalization of the problem of repairing integrity constraints in databases, and so it would be valuable to extend the work in this paper to query classes which can express universal quantification in limited ways, in order to build a more general theory of consistency repairs. Furthermore, while we chose to emphasize minimizing cardinality repairs in this paper, it is certainly worthwhile to investigate other desired constraints upon the repairs computed, such as minimizing side effects on the answer to the query. Another avenue is to generalize the problem to allow an arbitrary number of tuples to appear in the input, all of which must be added to the answer to a query by the repair.

Another important avenue of future work is to consider the generalization of the missing answer repair problem to semiring semantics~\cite{GreenKT07,GraedelTan24}. In this generalization, facts in databases are not simply true or false, but are annotated with values in a fixed semiring. The result of evaluating queries propagates these annotations through simple semiring operations, yielding an output annotation for each output. This present paper can be seen as studying the special case of \emph{set semantics}, which are induced naturally by the \emph{Boolean semiring}.

\bibliographystyle{plainurl}
\bibliography{admin/bib.bib}

@book{abiteboul1995foundations,
    author = {Abiteboul, Serge and Hull, Richard and Vianu, Victor},
    title = {Foundations of Databases: The Logical Level},
    year = {1995},
    isbn = {0201537710},
    publisher = {Addison-Wesley Longman Publishing Co., Inc.},
    address = {USA},
    edition = {1st}
}

@book{papadimitriou2003computational,
  author = {Papadimitriou, Christos H.},
  isbn = {978-0-201-53082-7},
  pages = {I-XV, 1-523},
  publisher = {Addison-Wesley},
  title = {Computational complexity},
  year = 1994
}

@inproceedings{BunemanKT01,
  author="Buneman, Peter
  and Khanna, Sanjeev
  and Wang-Chiew, Tan",
  editor="Van den Bussche, Jan
  and Vianu, Victor",
  title="Why and Where: A Characterization of Data Provenance",
  booktitle="Database Theory --- ICDT 2001",
  year="2001",
  publisher="Springer Berlin Heidelberg",
  address="Berlin, Heidelberg",
  pages="316--330",
  isbn="978-3-540-44503-6",
  doi="10.1007/3-540-44503-X_20"
}

@article{arenas2003answer,
author = {Arenas, Marcelo and Bertossi, Leopoldo and Chomicki, Jan},
title = {Answer sets for consistent query answering in inconsistent databases},
year = {2003},
issue_date = {July 2003},
publisher = {Cambridge University Press},
address = {USA},
volume = {3},
number = {4},
issn = {1471-0684},
url = {https://doi.org/10.1017/S1471068403001832},
doi = {10.1017/S1471068403001832},
journal = {Theory Pract. Log. Program.},
month = jul,
pages = {393–424},
numpages = {32}
}

@inproceedings{Chandra1977optimal,
    author = {Chandra, Ashok K. and Merlin, Philip M.},
    title = {Optimal implementation of conjunctive queries in relational data bases},
    year = {1977},
    isbn = {9781450374095},
    publisher = {Association for Computing Machinery},
    address = {New York, NY, USA},
    url = {https://doi.org/10.1145/800105.803397},
    doi = {10.1145/800105.803397},
    booktitle = {Proceedings of the Ninth Annual ACM Symposium on Theory of Computing},
    pages = {77–90},
    numpages = {14},
    location = {Boulder, Colorado, USA},
    series = {STOC '77}
}

@ARTICLE{Miao2018complexity,
    author={Miao, Dongjing and Cai, Zhipeng and Li, Jianzhong},
    journal={IEEE Transactions on Knowledge and Data Engineering}, 
    title={On the Complexity of Bounded View Propagation for Conjunctive Queries}, 
    year={2018},
    volume={30},
    number={1},
    pages={115-127},
    doi={10.1109/TKDE.2017.2758361}
}

@article{chaudhuri1997equivalence,
    title = {On the Equivalence of Recursive and Nonrecursive Datalog Programs},
    journal = {Journal of Computer and System Sciences},
    volume = {54},
    number = {1},
    pages = {61-78},
    year = {1997},
    issn = {0022-0000},
    doi = {https://doi.org/10.1006/jcss.1997.1452},
    author = {Surajit Chaudhuri and Moshe Y Vardi}
}

@inproceedings{DinurS14,
    author = {Dinur, Irit and Steurer, David},
    title = {Analytical approach to parallel repetition},
    year = {2014},
    isbn = {9781450327107},
    publisher = {Association for Computing Machinery},
    address = {New York, NY, USA},
    url = {https://doi.org/10.1145/2591796.2591884},
    doi = {10.1145/2591796.2591884},
    booktitle = {Proceedings of the Forty-Sixth Annual ACM Symposium on Theory of Computing},
    pages = {624–633},
    numpages = {10},
    location = {New York, New York},
    series = {STOC '14}
}

@article{GasarchKR95,
    author = {Gasarch, William and Krentel, M. and Rappoport, K.},
    year = {1995},
    month = {11},
    pages = {},
    title = {Opt{P} as the normal behavior of {NP}-complete problems},
    volume = {28},
    journal = {Theory of Computing Systems / Mathematical Systems Theory - MST},
    doi = {10.1007/BF01204168}
}

@inproceedings{GreenKT07,
    author = "Green, T. and Karvounarakis, G. and Tannen, V.",
    booktitle = "Principles of Database Systems {PODS}",
    pages = "31--40",
    title = "Provenance Semirings",
    year = "2007",
    organization = "ACM",
    doi = "10.1145/1265530.1265535"
}

@ARTICLE{Miao2023deletion,
    author={Miao, Dongjing and Cai, Zhipeng and Li, Jianzhong},
    journal={IEEE Transactions on Knowledge and Data Engineering}, 
    title={Deletion Propagation Revisited for Multiple Key Preserving Views}, 
    year={2023},
    volume={35},
    number={3},
    pages={2445-2456},
    doi={10.1109/TKDE.2021.3110851}
}

@article{Immerman86,
    title = {Relational queries computable in polynomial time},
    journal = {Information and Control},
    volume = {68},
    number = {1},
    pages = {86-104},
    year = {1986},
    issn = {0019-9958},
    doi = {https://doi.org/10.1016/S0019-9958(86)80029-8},
    author = {Neil Immerman}
}

@article{Krentel88,
    title = {The complexity of optimization problems},
    journal = {Journal of Computer and System Sciences},
    volume = {36},
    number = {3},
    pages = {490-509},
    year = {1988},
    issn = {0022-0000},
    doi = {https://doi.org/10.1016/0022-0000(88)90039-6},
    author = {Mark W. Krentel}
}

@article{MeliouGHKMS10,
  author       = {Alexandra Meliou and
                  Wolfgang Gatterbauer and
                  Joseph Y. Halpern and
                  Christoph Koch and
                  Katherine F. Moore and
                  Dan Suciu},
  title        = {Causality in Databases},
  journal      = {{IEEE} Data Eng. Bull.},
  volume       = {33},
  number       = {3},
  pages        = {59--67},
  year         = {2010}
}

@inproceedings{MeliouGMS10,
  author       = {A.~Meliou and
                  W.~Gatterbauer and
                  K.~Moore and
                  D.~Suciu},
  title        = {{WHY} SO? or {WHY} NO? Functional Causality for Explaining Query Answers},
  booktitle    = {{VLDB} workshop on Management
                  of Uncertain Data {(MUD} 2010)},
  series       = {{CTIT} Workshop Proceedings Series},
  volume       = {{WP10-04}},
  pages        = {3--17},
  year         = {2010},
  bibsource    = {dblp computer science bibliography, https://dblp.org}
}

@InProceedings{Miao2016complexity,
    author="Miao, Dongjing
    and Cai, Zhipeng
    and Liu, Xianmin
    and Li, Jianzhong",
    editor="Dinh, Thang N.
    and Thai, My T.",
    title="On the Complexity of Insertion Propagation with Functional Dependency Constraints",
    booktitle="Computing and Combinatorics ",
    year="2016",
    publisher="Springer International Publishing",
    address="Cham",
    pages="623--632",
    isbn="978-3-319-42634-1",
    doi="10.1007/978-3-319-42634-1_50"
}

@article{Miao2020functional,
    title = {Functional dependency restricted insertion propagation},
    journal = {Theoretical Computer Science},
    volume = {819},
    pages = {1-8},
    year = {2020},
    note = {Computing and Combinatorics},
    issn = {0304-3975},
    doi = {https://doi.org/10.1016/j.tcs.2017.03.043},
    author = {Dongjing Miao and Zhipeng Cai and Xianmin Liu and Jianzhong Li}
}

@InProceedings{Miao2020results,
    author="Miao, Dongjing
    and Li, Jianzhong
    and Cai, Zhipeng",
    editor="Zhang, Zhao
    and Li, Wei
    and Du, Ding-Zhu",
    title="New Results on the Complexity of Deletion Propagation",
    booktitle="Algorithmic Aspects in Information and Management",
    year="2020",
    publisher="Springer International Publishing",
    address="Cham",
    pages="336--345",
    isbn="978-3-030-57602-8",
    doi="10.1007/978-3-030-57602-8_30"
}

@book{OrponenM87,
  title={On approximation preserving reductions: complete problems and robust measures},
  author={Orponen, Pekka and Mannila, Heikki},
  year={1987},
  publisher={University of Helsinki}
}

@article{PapadimitriouY91,
    title = {Optimization, approximation, and complexity classes},
    journal = {Journal of Computer and System Sciences},
    volume = {43},
    number = {3},
    pages = {425-440},
    year = {1991},
    issn = {0022-0000},
    doi = {https://doi.org/10.1016/0022-0000(91)90023-X},
    author = {Christos H. Papadimitriou and Mihalis Yannakakis}
}

@article{Trakhtenbrot50,
     author       = {Boris Trakhtenbrot},
     year         = {1950},
     title        = {Impossibility of an algorithm for the decision problem on finite classes},
     journal      = {Doklady Akademii Nauk SSSR},
     volume       = {70},
     pages        = {569--572},
}

@article{DayalB82,
  author       = {Umeshwar Dayal and
                  Philip A. Bernstein},
  title        = {On the Correct Translation of Update Operations on Relational Views},
  journal      = {{ACM} Trans. Database Syst.},
  volume       = {7},
  number       = {3},
  pages        = {381--416},
  year         = {1982},
  url          = {https://doi.org/10.1145/319732.319740},
  doi          = {10.1145/319732.319740},
  bibsource    = {dblp computer science bibliography, https://dblp.org}
}

@article{bancilhon1981update,
author = {Bancilhon, F. and Spyratos, N.},
title = {Update semantics of relational views},
year = {1981},
issue_date = {Dec. 1981},
publisher = {Association for Computing Machinery},
address = {New York, NY, USA},
volume = {6},
number = {4},
issn = {0362-5915},
url = {https://doi.org/10.1145/319628.319634},
doi = {10.1145/319628.319634},
journal = {ACM Trans. Database Syst.},
month = dec,
pages = {557–575},
numpages = {19}
}

@article{cosmadakis1984updates,
  author = {Cosmadakis, Stavros S. and Papadimitriou, Christos H.},
  title = {Updates of Relational Views},
  year = {1984},
  issue_date = {Oct. 1984},
  publisher = {Association for Computing Machinery},
  address = {New York, NY, USA},
  volume = {31},
  number = {4},
  issn = {0004-5411},
  doi = {10.1145/1634.1887},
  journal = {J. ACM},
  month = sep,
  pages = {742–760},
  numpages = {19}
}

@inproceedings{BunemanKT02,
  author = {Buneman, Peter and Khanna, Sanjeev and Tan, Wang-Chiew},
  title = {On propagation of deletions and annotations through views},
  year = {2002},
  isbn = {1581135076},
  publisher = {Association for Computing Machinery},
  address = {New York, NY, USA},
  doi = {10.1145/543613.543633},
  booktitle = {Proceedings of the Twenty-First ACM SIGMOD-SIGACT-SIGART Symposium on Principles of Database Systems},
  pages = {150–158},
  numpages = {9},
  location = {Madison, Wisconsin},
  series = {PODS '02}
}

@article{FreireGIM15,
  author       = {Cibele Freire and
                  Wolfgang Gatterbauer and
                  Neil Immerman and
                  Alexandra Meliou},
  title        = {The Complexity of Resilience and Responsibility for Self-Join-Free
                  Conjunctive Queries},
  journal      = {Proc. {VLDB} Endow.},
  volume       = {9},
  number       = {3},
  pages        = {180--191},
  year         = {2015},
  doi          = {10.14778/2850583.2850592},
  bibsource    = {dblp computer science bibliography, https://dblp.org}
}

@inproceedings{FreireGIM20,
  author       = {Cibele Freire and
                  Wolfgang Gatterbauer and
                  Neil Immerman and
                  Alexandra Meliou},
  editor       = {Dan Suciu and
                  Yufei Tao and
                  Zhewei Wei},
  title        = {New Results for the Complexity of Resilience for Binary Conjunctive
                  Queries with Self-Joins},
  booktitle    = {Proceedings of the 39th {ACM} {SIGMOD-SIGACT-SIGAI} Symposium on Principles
                  of Database Systems, {PODS} 2020, Portland, OR, USA, June 14-19, 2020},
  pages        = {271--284},
  publisher    = {{ACM}},
  year         = {2020},
  url          = {https://doi.org/10.1145/3375395.3387647},
  doi          = {10.1145/3375395.3387647},
  bibsource    = {dblp computer science bibliography, https://dblp.org}
}

@InProceedings{Tan2004containment,
    author="Tan, Wang-Chiew",
    editor="Lausen, Georg
    and Suciu, Dan",
    title="Containment of Relational Queries with Annotation Propagation",
    booktitle="Database Programming Languages",
    year="2004",
    publisher="Springer Berlin Heidelberg",
    address="Berlin, Heidelberg",
    pages="37--53",
    isbn="978-3-540-24607-7",
    doi="10.1007/978-3-540-24607-7_4"
}

@inproceedings{Cong2006annotation,
    author = {Cong, Gao and Fan, Wenfei and Geerts, Floris},
    title = {Annotation propagation revisited for key preserving views},
    year = {2006},
    isbn = {1595934332},
    publisher = {Association for Computing Machinery},
    address = {New York, NY, USA},
    url = {https://doi.org/10.1145/1183614.1183705},
    doi = {10.1145/1183614.1183705},
    booktitle = {Proceedings of the 15th ACM International Conference on Information and Knowledge Management},
    pages = {632–641},
    numpages = {10},
    location = {Arlington, Virginia, USA},
    series = {CIKM '06}
}

@article{CongFGLL12,
  author       = {Gao Cong and
                  Wenfei Fan and
                  Floris Geerts and
                  Jianzhong Li and
                  Jizhou Luo},
  title        = {On the Complexity of View Update Analysis and Its Application to Annotation
                  Propagation},
  journal      = {{IEEE} Trans. Knowl. Data Eng.},
  volume       = {24},
  number       = {3},
  pages        = {506--519},
  year         = {2012},
  url          = {https://doi.org/10.1109/TKDE.2011.27},
  doi          = {10.1109/TKDE.2011.27},
  bibsource    = {dblp computer science bibliography, https://dblp.org}
}

@inproceedings{Kimelfeld12,
  author       = {Benny Kimelfeld},
  editor       = {Michael Benedikt and
                  Markus Kr{\"{o}}tzsch and
                  Maurizio Lenzerini},
  title        = {A dichotomy in the complexity of deletion propagation with functional
                  dependencies},
  booktitle    = {Proceedings of the 31st {ACM} {SIGMOD-SIGACT-SIGART} Symposium on
                  Principles of Database Systems, {PODS} 2012, Scottsdale, AZ, USA,
                  May 20-24, 2012},
  pages        = {191--202},
  publisher    = {{ACM}},
  year         = {2012},
  url          = {https://doi.org/10.1145/2213556.2213584},
  doi          = {10.1145/2213556.2213584},
  bibsource    = {dblp computer science bibliography, https://dblp.org}
}

@article{KimelfeldVW12,
  author       = {Benny Kimelfeld and
                  Jan Vondr{\'{a}}k and
                  Ryan Williams},
  title        = {Maximizing Conjunctive Views in Deletion Propagation},
  journal      = {{ACM} Trans. Database Syst.},
  volume       = {37},
  number       = {4},
  pages        = {24:1--24:37},
  year         = {2012},
  url          = {https://doi.org/10.1145/2389241.2389243},
  doi          = {10.1145/2389241.2389243},
  bibsource    = {dblp computer science bibliography, https://dblp.org}
}

@incollection{GraedelTan24,
  author="Gr{\"a}del, Erich
  and Tannen, Val",
  editor="Meer, Klaus
  and Rabinovich, Alexander
  and Ravve, Elena
  and Villaveces, Andr{\'e}s",
  title="Provenance Analysis and Semiring Semantics for First-Order Logic",
  bookTitle="Model Theory, Computer Science, and Graph Polynomials: Festschrift in Honor of Johann A. Makowsky",
  year="2025",
  publisher="Springer Nature Switzerland",
  address="Cham",
  pages="351--401",
  isbn="978-3-031-86319-6",
  doi="10.1007/978-3-031-86319-6_21",
}

@article{XuZhaAlaTan18,
  author = "Xu, J. and Zhang, W. and Alawini, A. and Tannen, V.",
  journal = "{IEEE} Data Eng. Bull.",
  number = "1",
  pages = "39-50",
  title = "Provenance Analysis for Missing Answers and Integrity Repairs",
  volume = "41",
  year = "2018",
  url = "http://sites.computer.org/debull/A18mar/p39.pdf"
}

@article{Shmueli93,
  author       = {Oded Shmueli},
  title        = {Equivalence of {DATALOG} Queries is Undecidable},
  journal      = {J. Log. Program.},
  volume       = {15},
  number       = {3},
  pages        = {231--241},
  year         = {1993},
  url          = {https://doi.org/10.1016/0743-1066(93)90040-N},
  doi          = {10.1016/0743-1066(93)90040-N},
  bibsource    = {dblp computer science bibliography, https://dblp.org}
}

@article{Grohe2009complexity,
  title      = {The Complexity of Datalog on Linear Orders},
  author     = {Martin Grohe and Goetz Schwandtner},
  doi        = {10.2168/LMCS-5(1:4)2009},
  journal    = {Logical Methods in Computer Science},
  issn       = {1860-5974},
  volume     = {Volume 5, Issue 1},
  eid        = 4,
  year       = {2009},
  month      = {Feb},
}

@inproceedings{levy1993equivalence,
    author = {Levy, Alon and Mumick, Inderpal Singh and Sagiv, Yehoshua and Shmueli, Oded},
    title = {Equivalence, query-reachability and satisfiability in Datalog extensions},
    year = {1993},
    isbn = {0897915933},
    publisher = {Association for Computing Machinery},
    address = {New York, NY, USA},
    url = {https://doi.org/10.1145/153850.153860},
    doi = {10.1145/153850.153860},
    booktitle = {Proceedings of the Twelfth ACM SIGACT-SIGMOD-SIGART Symposium on Principles of Database Systems},
    pages = {109–122},
    numpages = {14},
    location = {Washington, D.C., USA},
    series = {PODS '93}
}

@article{ArenasBC99,
author = {Arenas, Marcelo and Bertossi, Leopoldo and Chomicki, Jan},
title = {Answer sets for consistent query answering in inconsistent databases},
year = {2003},
issue_date = {July 2003},
publisher = {Cambridge University Press},
address = {USA},
volume = {3},
number = {4},
issn = {1471-0684},
url = {https://doi.org/10.1017/S1471068403001832},
doi = {10.1017/S1471068403001832},
journal = {Theory Pract. Log. Program.},
month = jul,
pages = {393–424},
numpages = {32}
}

@inproceedings{ChapmanJ09,
  author       = {Adriane Chapman and
                  H. V. Jagadish},
  editor       = {Ugur {\c{C}}etintemel and
                  Stanley B. Zdonik and
                  Donald Kossmann and
                  Nesime Tatbul},
  title        = {Why not?},
  booktitle    = {Proceedings of the {ACM} {SIGMOD} International Conference on Management
                  of Data, {SIGMOD} 2009, Providence, Rhode Island, USA, June 29 - July
                  2, 2009},
  pages        = {523--534},
  publisher    = {{ACM}},
  year         = {2009},
  url          = {https://doi.org/10.1145/1559845.1559901},
  doi          = {10.1145/1559845.1559901},
  bibsource    = {dblp computer science bibliography, https://dblp.org}
}

@article{HerschelHT09,
  author       = {Melanie Herschel and
                  Mauricio A. Hern{\'{a}}ndez and
                  Wang Chiew Tan},
  title        = {Artemis: {A} System for Analyzing Missing Answers},
  journal      = {Proc. {VLDB} Endow.},
  volume       = {2},
  number       = {2},
  pages        = {1550--1553},
  year         = {2009},
  doi          = {10.14778/1687553.1687588},
  bibsource    = {dblp computer science bibliography, https://dblp.org}
}

@article{HerschelH10,
  author       = {Melanie Herschel and
                  Mauricio A. Hern{\'{a}}ndez},
  title        = {Explaining Missing Answers to {SPJUA} Queries},
  journal      = {Proc. {VLDB} Endow.},
  volume       = {3},
  number       = {1},
  pages        = {185--196},
  year         = {2010},
  doi          = {10.14778/1920841.1920869},
  bibsource    = {dblp computer science bibliography, https://dblp.org}
}

@inproceedings{BidoitHT14a,
  author       = {Nicole Bidoit and
                  Melanie Herschel and
                  Katerina Tzompanaki},
  editor       = {Sihem Amer{-}Yahia and
                  Vassilis Christophides and
                  Anastasios Kementsietsidis and
                  Minos N. Garofalakis and
                  Stratos Idreos and
                  Vincent Leroy},
  title        = {Query-Based Why-Not Provenance with {N}ed{E}xplain},
  booktitle    = {Proceedings of the 17th International Conference on Extending Database
                  Technology, {EDBT} 2014, Athens, Greece, March 24-28, 2014},
  pages        = {145--156},
  publisher    = {OpenProceedings.org},
  year         = {2014},
  doi          = {10.5441/002/EDBT.2014.14},
  bibsource    = {dblp computer science bibliography, https://dblp.org}
}

@inproceedings{BidoitHT14b,
  author       = {Nicole Bidoit and
                  Melanie Herschel and
                  Katerina Tzompanaki},
  editor       = {Adriane Chapman and
                  Bertram Lud{\"{a}}scher and
                  Andreas Schreiber},
  title        = {Immutably Answering Why-Not Questions for Equivalent Conjunctive Queries},
  booktitle    = {6th Workshop on the Theory and Practice of Provenance, TaPP'14, Cologne,
                  Germany, June 12-13, 2014},
  publisher    = {{USENIX} Association},
  year         = {2014},
  url          = {https://www.usenix.org/conference/tapp2014/agenda/presentation/bidoit},
  bibsource    = {dblp computer science bibliography, https://dblp.org}
}

@article{HuangCDN08,
  author       = {Jiansheng Huang and
                  Ting Chen and
                  AnHai Doan and
                  Jeffrey F. Naughton},
  title        = {On the provenance of non-answers to queries over extracted data},
  journal      = {Proc. {VLDB} Endow.},
  volume       = {1},
  number       = {1},
  pages        = {736--747},
  year         = {2008},
  doi          = {10.14778/1453856.1453936},
  bibsource    = {dblp computer science bibliography, https://dblp.org}
}

@inproceedings{TranC10,
  author       = {Quoc Trung Tran and
                  Chee{-}Yong Chan},
  editor       = {Ahmed K. Elmagarmid and
                  Divyakant Agrawal},
  title        = {How to ConQueR why-not questions},
  booktitle    = {Proceedings of the {ACM} {SIGMOD} International Conference on Management
                  of Data, {SIGMOD} 2010, Indianapolis, Indiana, USA, June 6-10, 2010},
  pages        = {15--26},
  publisher    = {{ACM}},
  year         = {2010},
  url          = {https://doi.org/10.1145/1807167.1807172},
  doi          = {10.1145/1807167.1807172},
  bibsource    = {dblp computer science bibliography, https://dblp.org}
}

@article{MakhijaG23,
  author       = {Neha Makhija and
                  Wolfgang Gatterbauer},
  title        = {A Unified Approach for Resilience and Causal Responsibility with Integer
                  Linear Programming {(ILP)} and {LP} Relaxations},
  journal      = {Proc. {ACM} Manag. Data},
  volume       = {1},
  number       = {4},
  pages        = {228:1--228:27},
  year         = {2023},
  url          = {https://doi.org/10.1145/3626715},
  doi          = {10.1145/3626715},
  bibsource    = {dblp computer science bibliography, https://dblp.org}
}

@InProceedings{Miao2018aggregation,
    author="Miao, Dongjing
    and Cai, Zhipeng",
    editor="Kim, Donghyun
    and Uma, R. N.
    and Zelikovsky, Alexander",
    title="On the Complexity of Resilience for Aggregation Queries",
    booktitle="Combinatorial Optimization and Applications",
    year="2018",
    publisher="Springer International Publishing",
    address="Cham",
    pages="696--706",
    isbn="978-3-030-04651-4",
    doi="10.1007/978-3-030-04651-4_47"
}

@article{WuM13,
  author       = {Eugene Wu and
                  Samuel Madden},
  title        = {Scorpion: Explaining Away Outliers in Aggregate Queries},
  journal      = {Proc. {VLDB} Endow.},
  volume       = {6},
  number       = {8},
  pages        = {553--564},
  year         = {2013},
  doi          = {10.14778/2536354.2536356},
  bibsource    = {dblp computer science bibliography, https://dblp.org}
}

@inproceedings{RoyS14,
  author       = {Sudeepa Roy and
                  Dan Suciu},
  editor       = {Curtis E. Dyreson and
                  Feifei Li and
                  M. Tamer {\"{O}}zsu},
  title        = {A formal approach to finding explanations for database queries},
  booktitle    = {International Conference on Management of Data, {SIGMOD} 2014, Snowbird,
                  UT, USA, June 22-27, 2014},
  pages        = {1579--1590},
  publisher    = {{ACM}},
  year         = {2014},
  url          = {https://doi.org/10.1145/2588555.2588578},
  doi          = {10.1145/2588555.2588578},
  bibsource    = {dblp computer science bibliography, https://dblp.org}
}

@article{RoyOS15,
  author       = {Sudeepa Roy and
                  Laurel J. Orr and
                  Dan Suciu},
  title        = {Explaining Query Answers with Explanation-Ready Databases},
  journal      = {Proc. {VLDB} Endow.},
  volume       = {9},
  number       = {4},
  pages        = {348--359},
  year         = {2015},
  doi          = {10.14778/2856318.2856329},
  bibsource    = {dblp computer science bibliography, https://dblp.org}
}

@inproceedings{DantsinEGV97,
  author       = {Evgeny Dantsin and
                  Thomas Eiter and
                  Georg Gottlob and
                  Andrei Voronkov},
  title        = {Complexity and Expressive Power of Logic Programming},
  booktitle    = {Proceedings of the Twelfth Annual {IEEE} Conference on Computational
                  Complexity, Ulm, Germany, June 24-27, 1997},
  pages        = {82--101},
  publisher    = {{IEEE} Computer Society},
  year         = {1997},
  url          = {https://doi.org/10.1109/CCC.1997.612304},
  doi          = {10.1109/CCC.1997.612304},
  bibsource    = {dblp computer science bibliography, https://dblp.org}
}

@InProceedings{lopatenko06complexity,
  author="Lopatenko, Andrei
  and Bertossi, Leopoldo",
  editor="Schwentick, Thomas
  and Suciu, Dan",
  title="Complexity of Consistent Query Answering in Databases Under Cardinality-Based and Incremental Repair Semantics",
  booktitle="Database Theory -- ICDT 2007",
  year="2006",
  publisher="Springer Berlin Heidelberg",
  address="Berlin, Heidelberg",
  pages="179--193",
  isbn="978-3-540-69270-6",
  doi="10.1007/11965893_13"
}

@InProceedings{Bertossi2018measuring,
  author="Bertossi, Leopoldo",
  editor="Ciucci, Davide
  and Pasi, Gabriella
  and Vantaggi, Barbara",
  title="Measuring and Computing Database Inconsistency via Repairs",
  booktitle="Scalable Uncertainty Management",
  year="2018",
  publisher="Springer International Publishing",
  address="Cham",
  pages="368--372",
  isbn="978-3-030-00461-3",
  doi="10.1007/978-3-030-00461-3_26"
}

@book{Fan2012foundations,
    author = {Fan, Wenfei and Geerts, Floris},
    title = {Foundations of Data Quality Management},
    year = {2012},
    isbn = {160845777X},
    publisher = {Morgan \& Claypool Publishers},
}

@inproceedings{Kolahi2009approximating,
    author = {Kolahi, Solmaz and Lakshmanan, Laks V. S.},
    title = {On approximating optimum repairs for functional dependency violations},
    year = {2009},
    isbn = {9781605584232},
    publisher = {Association for Computing Machinery},
    address = {New York, NY, USA},
    url = {https://doi.org/10.1145/1514894.1514901},
    doi = {10.1145/1514894.1514901},
    booktitle = {Proceedings of the 12th International Conference on Database Theory},
    pages = {53–62},
    numpages = {10},
    location = {St. Petersburg, Russia},
    series = {ICDT '09}
}

@article{Carmeli2024database,
  author = {Carmeli, Nofar and Grohe, Martin and Kimelfeld, Benny and Livshits, Ester and Tibi, Muhammad},
  title = {Database Repairing with Soft Functional Dependencies},
  year = {2024},
  issue_date = {June 2024},
  publisher = {Association for Computing Machinery},
  address = {New York, NY, USA},
  volume = {49},
  number = {2},
  issn = {0362-5915},
  doi = {10.1145/3651156},
  journal = {ACM Trans. Database Syst.},
  month = apr,
  articleno = {8},
  numpages = {34}
}

@inproceedings{Vardi82,
  author       = {Moshe Y. Vardi},
  editor       = {Harry R. Lewis and
                  Barbara B. Simons and
                  Walter A. Burkhard and
                  Lawrence H. Landweber},
  title        = {The Complexity of Relational Query Languages (Extended Abstract)},
  booktitle    = {Proceedings of the 14th Annual {ACM} Symposium on Theory of Computing,
                  May 5-7, 1982, San Francisco, California, {USA}},
  pages        = {137--146},
  publisher    = {{ACM}},
  year         = {1982},
  url          = {https://doi.org/10.1145/800070.802186},
  doi          = {10.1145/800070.802186},
  bibsource    = {dblp computer science bibliography, https://dblp.org}
}

\appendix
\section{Appendix}
\label{sec:app}
Recall that we have defined $\DL$ as programs including only positive literals. In this section, we show $\EXP$-hardness of $\dbSAT{\DL}$ by a polynomial-time reduction from the problem of the Datalog non-emptiness problem for the two-element linear order. To the author's knowledge, this lower bound for $\dbSAT{\DL}$ is an unpublished folklore result.

\begin{definition}[The Datalog non-emptiness problem for database instance $\bfI$]
Let $\bfI$ denote a fixed database instance. The \emph{Datalog non-emptiness problem} over $\bfI$ is to determine, given a $\DL$ program $P$ (over the EDB schema of $\bfI$), whether or not $P(\bfI)$ is nonempty.
\end{definition}

\begin{proposition}[\cite{Grohe2009complexity}]
The datalog non-emptiness problem over $\bfI_2 = \{ R(a,b) \}$ (i.e., the database representing a $2$-element linear order) is $\EXP$-hard.
\end{proposition}

\begin{theorem}
$\dbSAT{\DL}$ is $\EXP$-hard.
\end{theorem}
\begin{proof}
We may assume without loss of generality that the set $\Sigma$ of EDBs contains at least one relation symbol (otherwise, no safe programs can be written), and that the binary relation symbol $R$ does not occur in $\Sigma$. Given a rule $r$ over the schema $\{ R \}$, let $\pi(r)$ be the rule obtained by the following transformations:
\begin{enumerate}
\item for each equality atom $s = t$ occurring in $r$, we place $s = t$ in $\pi(r)$,
\item for each atom $R(s,t)$ in $r$, we place an atom $S(s,t)$ in $\pi(r)$, and
\item for each variable $z$ occurring in $r$, we place an atom $V(z)$ in $\pi(r)$.
\end{enumerate}
Let $P$ denote an arbitrary datalog program over the schema $\{ R \}$ with rules $r_1,\hdots,r_m$ and answer predicate $\ans_P(\tup{x})$. We define a program $\widehat{P}$ over the schema $\{ U \}$ with answer predicate $\ans_{\widehat{P}}$ which has rules $\pi(r_i)$ for $i \leq m$, as well as the rules
\begin{align}
\label{eq:dbsatDL-rules}
V(x) &\df x = a, U(x,\hdots,x), \\
V(x) &\df x = b, U(x,\hdots,x), \\
S(x,y) &\df x = a, y = b, U(x,\hdots,x), U(y,\hdots,y), \\
\ans_{\widehat{P}} &\df U(a,\hdots,a), U(b,\hdots,b), \ans_P(\tup{x}),
\end{align}
respectively, where $V$, $S$, and $\ans_{\widehat{P}}$ are IDBs not occurring in $P$. Observe that, for any database instance $\bfI$, if $\widehat{P}(\bfI) \neq \emptyset$, then $U(a,\hdots,a), U(b,\hdots,b) \in \bfI$, which further implies that $S(\bfI) = \{ (a,b) \}$ and $V(\bfI) = \{ a, b \}$. Clearly the map $P \mapsto \widehat{P}$ is computable in polynomial time, so it remains only to show that $P(\bfI_2) \neq \emptyset$ if and only if $\widehat{P}$ is satisfiable.

For the ``only if'' direction, suppose $P(\bfI_2) \neq \emptyset$. Then there exists some conjunctive query $Q_1 \in \Unfoldings(P)$ such that $Q_1(\bfI_2)$ is non-empty. Let $g$ be a satisfying assignment for $Q_1$ in $\bfI_2$. By the safety condition, we may assume that $\rng(g) \subseteq \{ a,b \}$. Observe that $\pi(Q_1)$ is in $\widehat{P}^\infty$, and note that $\pi(Q_1)$ does not contain any atoms of the form $V(t)$ where $t$ is a constant. Let $Q'_1$ be the CQ, defined by a rule $r'$, obtained from $\pi(Q_1)$ by replacing each atom $S(s,t)$ with the atoms $s = a, t = b, U(s,\hdots,s), U(t,\hdots,t)$, and replacing each atom $V(z)$ in $\pi(Q_1)$ with the facts $z = g(z), U(z,\hdots,z)$. Clearly $Q'_1 \in \widehat{P}^\infty$ and contains only EDB predicates. Finally, let $Q''_1$ be the CQ defined by the rule $r''$ given by
\[
\ans_{\widehat{P}} \df U(a,\hdots,a), U(b,\hdots,b), \Body_{r'}.
\]
Clearly $Q''_1 \in \Unfoldings(\widehat{P})$, and $g$ is a satisfying assignment for $Q''_1$ in $\bfI$. Hence $\widehat{P}(\bfI) \neq \emptyset$, and so $\widehat{P}$ is satisfiable.

For the ``if'' direction, suppose that $\widehat{P}$ is satisfiable. Then there exists some $\bfI$ such that $\widehat{P}(\bfI) \neq \emptyset$. Hence there exists some $Q_2 \in \Unfoldings(\widehat{P})$ such that $Q_2(\bfI)$ is non-empty. Let $g$ be a satisfying assignment for $Q_2$ in $\bfI$. By the construction of $\widehat{P}$, there must exist a CQ $Q'_2$ in $\widehat{P}^\infty$ containing $V$ and $S$ predicates such that $Q'_2$ is obtained by replacing these predicates with appropriate substitutions for the bodies of the rules (2)-(4). This implies that there exists some $Q \in \Unfoldings(P)$ such that $\pi(Q) = Q'_2$, and $g$ is a satisfying assignment for $Q$ in $\bfI_2$. Hence $P(\bfI_2)$ is nonempty.
\end{proof}

\end{document}